\documentclass[12pt]{article}
\usepackage[utf8]{inputenc}
\usepackage{natbib}
\usepackage{blindtext}
\usepackage[T1]{fontenc}
\usepackage{amsfonts}
\usepackage[english]{babel}
\usepackage{amsthm}
\usepackage{amssymb}
\usepackage{relsize}
\usepackage{amsmath}
\usepackage{mathrsfs} 
\usepackage[parfill]{parskip}
\usepackage{bbm}
\usepackage{graphicx}
\usepackage{enumerate}
\usepackage{mathtools}
\usepackage{xcolor}
\usepackage[pdftex , colorlinks = true , citecolor = red , linkcolor = blue ]{hyperref}	
\usepackage{array}
\usepackage{booktabs}
\usepackage{dsfont}
\usepackage{subcaption}
\usepackage{multirow}
\usepackage{longtable}
\usepackage{bm}
\usepackage{txfonts}
\usepackage{bbm}

\newcolumntype{P}[1]{>{\centering\arraybackslash}p{#1}}
\newcolumntype{M}[1]{>{\centering\arraybackslash}m{#1}}
\newtheorem{assumption}{Assumption}
\newtheorem{lemma}{Lemma}
\newtheorem{theorem}{Theorem}
\newtheorem{proposition}{Proposition}

\newcommand{\va}{\mathrm{Var}}
\newcommand{\e}{\mathrm{E}}
\newcommand{\cov}{\mathrm{Cov}}
\newcommand{\tr}{\mathrm{tr}}

\newcommand\indep{\protect\mathpalette{\protect\independenT}{\perp}}
\def\independenT#1#2{\mathrel{\rlap{$#1#2$}\mkern2mu{#1#2}}}
\DeclareMathOperator*{\argmin}{argmin}
\DeclareMathOperator*{\sgn}{sgn}
\DeclareMathOperator*{\diag}{diag}

\numberwithin{equation}{section}
\allowdisplaybreaks

\begin{document}
	\title{\Large{\textbf{Simultaneous Estimation of Graphical Models by Neighborhood Selection}}}
	\author{\bigskip
		{Ilias Moysidis and Bing Li}\\
		{\em Pennsylvania State University}
	}
	\date{}
	
	\maketitle

	\begin{abstract}
		In many applications concerning statistical graphical models the data originate from several subpopulations that share similarities but have also significant differences. This raises the question of how to estimate several graphical models simultaneously. Compiling all the data together to estimate a single graph would ignore the differences among subpopulations. On the other hand, estimating a graph from each subpopulation separately does not make efficient use of the common structure in the data. We develop a new method for simultaneous estimation of multiple graphical models by estimating the topological neighborhoods of the involved variables under a sparse inducing penalty that takes into account the common structure in the subpopulations. Unlike the existing methods for joint graphical models, our method does not rely on spectral decomposition of large matrices, and is therefore more computationally attractive for estimating large networks. In addition, we develop the asymptotic properties of our method, demonstrate its the numerical complexity, and compare it with several existing methods by simulation. Finally, we apply our method to the estimation of genomic networks for a lung cancer dataset which consists of several subpopulations.
		\bigskip
		
	\end{abstract}

	\section{Introduction}
	Graphical models are a useful tool for constructing statistical networks for a wide range of applications such as speech recognition, computer vision and genomics. One recent focus in this research is the simultaneous estimation of multiple graphs from several subpopulations. The current approach to this problem is by modifying the graphical lasso \citep{yuan2007model} to take into account of the common structure in the subpopulations \citep{guo2011joint,danaher2014joint}. However, this involves the spectral decomposition of large matrices and is computationally infeasible for some biological applications. In this paper we propose to generalize the neighborhood selection \citep{meinshausen2006high} to the simultaneous estimation problem, which saves substantial amount of computing time and can be applied to large genetic networks.

	Consider a $p$-dimensional random vector $\mathbf{X}=(X_{1},\ldots,X_{p})^{\intercal}$. The graphical model of $\mathbf{X}$ is represented by an undirected graph $\mathcal{G}=(\mathcal{V},\mathcal{E})$, where $\mathcal{V}=\{1,\ldots,p\}$ is the set of vertices and $\mathcal{E}=\{(i,j):i,j\in\mathcal{V},i\ne j\}$ is the set of edges. Because for an undirected graph $(i,j)$ and $(j,i)$ represent the same edge, we assume $i<j$ for $(i,j)\in\mathcal{E}$ without loss of generality. In a statistical graphical model, the set of edges $\mathcal{E}$ is defined by the relation
	\begin{align}
		\label{rel:1}
		(i,j)\notin\mathcal{E}\quad\Leftrightarrow\quad X_{i}\indep X_{j}|\mathbf{X}_{-\{i,j\}},
	\end{align}	
	where $\mathbf{X}_{-\{i,j\}}$ is the vector $\mathbf{X}$ with its $i$-th and $j$-th components removed, that is $\mathbf{X}_{-\{i,j\}}=\left\{X_{k}:k\in\mathcal{V}\backslash\{i,j\}\right\}$, and the notation $A\indep B|C$ means $A$ and $B$ are conditionally independent given $C$. Of special interest is the case where $\mathbf{X}$ follows a multivariate Gaussian distribution $\mathcal{N}_{p}(\bm{0},\mathbf{\Sigma})$. Let $\mathbf{\Omega}=\mathbf{\Sigma}^{-1}$ be the precision matrix and $\omega_{ij}$ the $(i,j)$-th component of $\mathbf{\Omega}$. Under the Gaussian assumption, because of the relation
	\begin{align*}
		X_{i}\indep X_{j}|\mathbf{X}_{-\{i,j\}}\quad\Leftrightarrow\quad\omega_{ij}=0,
	\end{align*}
	the estimation of the edge set $\mathcal{E}$ is equivalent to the estimation of the positions of the zero components of the precision matrix. \cite{lauritzen1996graphical} and \cite{sinoquet2014probabilistic} give excellent reviews of the theory of statistical graphical models.
	
	The problem of estimating the zero components of the precision matrix dates back to \cite{dempster1972covariance}. He proposed a forward selection procedure in which he starts with the sample covariance matrix and gradually chooses elements of the inverse of the sample covariance matrix to go to zero. The procedure is stopped when adding a new zero component does not change the fit significantly. However, the overall error properties of the step-wise procedure are hard to study. \cite{drton2004model} proposed a joint hypothesis testing under the same significance level for all the pairs of nodes.
	
	\cite{meinshausen2006high} developed the neighborhood selection method that uses an equivalent relationship to \eqref{rel:1} to define the edges of the graph. They applied penalized regressions to find the neighborhood of each node and assemble these neighborhoods together to recover the entire graph. Up until that point, graph and precision matrix estimation were separate problems. \cite{yuan2007model} proposed a penalized likelihood method, called the graphical lasso, to estimate the precision matrix that forces its small components to zero. \cite{rocha2008path} developed a method that allows for joint estimation of all the neighborhoods in the graph that is based on a minimization of a penalized pseudo-likelihood. Other important recent advances include \cite{liu2009nonparanormal,liu2012high} and \cite{xue2012regularized}, which considered the non-Gaussian case.
	
	The simultaneous estimation problem was first considered by \cite{guo2011joint}, which assumes that the data come from different subpopulations that share a common graph structure but also have significant differences. Under such circumstances, to estimate each graph separately would be to waste information in uncovering the common structure. On the other hand, estimating a single graph by merging data from all the subpopulations together would ignore their differences. \cite{guo2011joint} proposed a simultaneous estimation procedure by decomposing each precision matrix into two components: one common to all subpopulations and another specific to each subpopulation. More recently, \cite{danaher2014joint} introduced a different simultaneous estimation method based on the fused graphical lasso  and group graphical lasso. By making some additional assumptions for the common structure, they overcame the problem of the non-convex penalty, resulting in a more computationally efficient method. 
	
	Our proposed method, which we call Simultaneous Neighborhood Selection (SNS), has the following advantages over the above methods. First, compared with \cite{guo2011joint}, instead of trying to adapt graphical lasso to simultaneous estimation, we adapt neighborhood selection to simultaneous estimation by introducing an extra penalty term that enforces the common structure among the subpopulations. We further simplify the penalty by local linear approximation \citep{zou2008one} which reduces the procedure to that of adaptive lasso. This enables us to bypass the large eigenvalue decomposition. Second, even though the method of \cite{danaher2014joint} is very fast computationally, it comes at a cost of making more assumptions about the structure of the graphs. In particular, they assume that the precision matrices of each subpopulation are block-diagonal. Our method does not need to make these additional assumptions.
	
	The rest of the paper is organized as follows. In section 2 we give an overview of the penalized neighborhood selection method. In section 3 we propose our new method for simultaneous estimation via neighborhood selection. In section 4 we develop the algorithm for finding the minimum of the objective function and demonstrate its computational complexity. In section 5 we compare the performance of SNS with existing methods on simulated datasets, both in terms of ROC curves and CPU time. In section 5 we apply SNS on a lung cancer dataset.

	\section{Methodology}\label{methodology}
	\subsection{Neighborhood Selection Method}
	We first give an overview of the neighborhood selection method for estimating a single graph. Suppose $\mathbf{X}=(X_{1},\ldots,X_{p})^{\intercal}$ is a random vector that follows a multivariate Gaussian distribution $\mathcal{N}_{p}(\bm{0},\mathbf{\Sigma})$, with precision matrix $\mathbf{\Omega}=\mathbf{\Sigma}^{-1}$. The neighborhood of a vertex $i\in\mathcal{V}$ is the smallest set $\mathcal{A}\subset\mathcal{V}$ such that $X_{i}\indep \mathbf{X}_{-\mathcal{A}\cup\{i\}}|\mathbf{X}_{\mathcal{A}}$, and is denoted by $\mathrm{ne}(i)$. Let $\mathcal{G}=(\mathcal{V},\mathcal{E})$ be the undirected graph with edge set $\mathcal{E}$ defined by the relation
	\begin{align}
		\label{rel:2}
		(i,j)\in\mathcal{E}\quad\Leftrightarrow\quad j\in\mathrm{ne}(i).
	\end{align}
	It can be shown that the edge sets determined by the relations \eqref{rel:1} and \eqref{rel:2} are identical \citep[Proposition C.5]{lauritzen1996graphical}. Since $\mathbf{X}\sim\mathcal{N}(\mathbf{0},\mathbf{\Omega}^{-1})$, where $\mathbf{\Omega}=(\omega_{ij})$, then for each $j\in\mathcal{V}$, the conditional distribution of $X_{j}|\mathbf{X}_{-j}$ is $\mathcal{N}\left(\sum_{i\ne j}\theta_{ij}X_{i},\sigma^{2}\right)$, where
		\begin{align*}
			\theta_{ij}=-\frac{\omega_{ij}}{\omega_{jj}}\quad,\quad\sigma^{2}=\va\left(X_{j}\right)-\cov\left(X_{j},\mathbf{X}_{-j}\right)\va\left(\mathbf{X}_{-j}\right)^{-1}\cov\left(\mathbf{X}_{-j},X_{j}\right).
		\end{align*}
		From this we can see that, under the assumption $\mathbf{X}\sim\mathcal{N}(\mathbf{0},\mathbf{\Omega}^{-1})$,
	\begin{align*}
		\theta_{ij}=0\quad\Leftrightarrow\quad\omega_{ij}=0\quad\Leftrightarrow\quad X_{i}\indep X_{j}|\mathbf{X}_{-\{i,j\}}    \quad\Leftrightarrow\quad (i,j)\notin\mathcal{E}.
	\end{align*}
	Thus, the edge set $\{(i,j):i<j,\,\omega_{ij}\ne 0\}$ is the same as the set $\{(i,j):i<j,\,\theta_{ij}\ne 0\}$, which is completely determined by the system of neighborhoods $\{\mathrm{ne}(i):i\in\mathcal{V}\}$. In this way, estimating the graph reduces to neighborhood selection.
	
	Another way to represent the statement $X_{j}|\mathbf{X}_{-j}\sim\mathcal{N}\left(\sum_{i\ne j}\theta_{ij}X_{i},\sigma^{2}\right)$ is by the linear regression model
	\begin{align*}
		X_{j}=\sum_{i\ne j}\theta_{ij}X_{i}+\epsilon,
	\end{align*}
	where the error $\epsilon=X_{j}-\sum_{i\ne j}\theta_{ij}X_{i}$ follows a $\mathcal{N}(0,\sigma^{2})$ distribution and is independent of $\sum_{i\ne j}\theta_{ij}X_{i}$. Therefore, finding the graph $\mathcal{G}$ boils down to regressing each variable $X_{j}$ against the remaining $\mathbf{X}_{-j}$ and finding the zero coefficients $\theta_{ij}$.

	The neighborhood selection method for estimating $\mathcal{G}$ proceeds as follows. Define the matrix $\mathbf{\Theta}=(\bm{\theta}_{1},\ldots,\bm{\theta}_{p})$, where $\bm{\theta}_{j}=(\theta_{1j},\ldots,\theta_{pj})^{\intercal}$ such that $\theta_{jj}=0$ and $\theta_{lj}$ is the regression coefficient defined above for $l\ne j$. Suppose we observe an i.i.d. sample $\mathbf{X}_{1},\ldots,\mathbf{X}_{n}$, where $\mathbf{X}_{i}=(X_{i1},\ldots,X_{ip})^{\intercal}$. Let $\mathds{X}$ denote the matrix $\left(\mathbf{X}_{1}\ldots \mathbf{X}_{n}\right)^{\intercal}$ and 
	$\mathds{X}_{j}$ the $j$-th column of $\mathds{X}$. We perform linear regression of $\mathds{X}_{j}$ on $\{\mathds{X}_{1},\ldots,\mathds{X}_{p}\}\backslash\{\mathds{X}_{j}\}$ by minimizing the least squares criterion
	\begin{align*}
		\frac{1}{2n}\left|\left|\mathds{X}_{j}-\mathds{X}\bm{\theta}_{j}\right|\right|_{2}^{2},
	\end{align*}
	where $||\cdot||_{2}$ denotes the $\ell_{2}$ norm.
	To induce sparsity so that we can cover the $p>n$ case, we add an $\ell_{1}$-penalty  to produce the penalized estimator
	\begin{align}\label{indivopt}
		\hat{\bm{\theta}}_{j}=\argmin_{\bm{\theta}\in\mathbb{R}^{p},\theta_{jj}=0}\left(\frac{1}{2n}\left|\left|\mathds{X}_{j}-\mathds{X}\bm{\theta}\right|\right|_{2}^{2}+\lambda||\bm{\theta}||_{1}\right),\quad j=1,\ldots,p,
	\end{align}
	where $\lambda$ is a nonnegative tuning parameter that controls the degree of sparsity. We rewrite the $p$ equations in \eqref{indivopt} into a matrix form as
	\begin{align*}
		\hat{\mathbf{\Theta}}=\argmin_{\mathbf{\Theta}\in\mathbb{R}^{p\times p},\diag(\mathbf{\Theta})=\mathbf{0}}\left(\frac{1}{2n}||\mathds{X}(\mathbf{I}-\mathbf{\Theta})||_{F}^{2}+\lambda||\mathbf{\Theta}||_{1}\right),
	\end{align*}
	where $||\cdot||_{F}$ and $||\cdot||_{1}$ denote the Frobenius norm and the elementwise matrix $\ell_{1}$ norm, respectively, and $\diag(\mathbf{\Theta})=(\theta_{11},\ldots,\theta_{pp})^{\intercal}$.

	Let $\hat{\bm{\theta}}_{j}$ be the solution of \eqref{indivopt} and let $\hat{\mathrm{ne}}(i)=\{i:i\ne j,\,\hat{\theta}_{ij}\ne 0\}$. Because $j\in\hat{\mathrm{ne}}(i)$ does not imply $i\in\hat{\mathrm{ne}}(j)$ and vice versa, \cite{meinshausen2006high} proposed two estimators for the edge set $\mathcal{E}$: a conservative estimator $\hat{\mathcal{E}}=\{(i,j):j\in\hat{\mathrm{ne}}(i)\text{ and }i\in\hat{\mathrm{ne}}(j)\}$, and a liberal estimator $\hat{\mathcal{E}}=\{(i,j):j\in\hat{\mathrm{ne}}(i)\text{ or }i\in\hat{\mathrm{ne}}(j)\}$.

	\subsection{Simultaneous Neighborhood Selection}\label{SNS}	
	We now consider simultaneous estimation of multiple graphs from several subpopulations. We assume that there are $K$ different subpopulations whose graphs, though different, share a set of common edges. For each $k=1,\ldots,K$, suppose we observe an i.i.d. sample $\mathbf{X}_{1}^{(k)},\ldots,\mathbf{X}_{n_{k}}^{(k)}$, where $\mathbf{X}_{i}^{(k)}=\left(X_{i1}^{(k)},\ldots,X_{ip}^{(k)}\right)^{\intercal}$. We assume that $\mathbf{X}_{i}^{(k)}$ is distributed as $\mathcal{N}_{p}(\bm{0},\mathbf{\Sigma}^{(k)})$. Let $\mathds{X}^{(k)}$ denote the matrix $\left(\mathbf{X}_{1}^{(k)},\ldots,\mathbf{X}_{n_{k}}^{(k)}\right)^{\intercal}$ and $\mathbf{\Theta}^{(k)}$ the matrix of coefficients defined in section \ref{methodology} for each subpopulation $k$.

	To take advantage of the information across the subpopulations we reparameterize each $\mathbf{\Theta}^{(k)}$ as $\mathbf{H}\circ\mathbf{\Gamma}^{(k)}$, where $\circ$ is the Hadamard matrix product, $\mathbf{H}=(\eta_{lj})$ is a matrix common to all subpopulations and $\mathbf{\Gamma}^{(k)}=(\gamma_{lj}^{(k)})$ is a matrix specific to subpopulation $k$. To eliminate sign ambiguity we assume $\eta_{lj}\geq0$ for all $l$ and $j$, and to be consistent with the fact that $\theta_{jj}^{(k)}=0$ we set $\eta_{jj}=\gamma_{jj}^{(k)}=0$ for all $j,k$. In this reparameterization, the common factor $\eta_{lj}$ controls the presence of vertex $l$ in the neighborhood of $j$ in all of the graphs, and $\gamma_{lj}^{(k)}$ accommodates the differences in the neighborhood of $j$ between individual graphs. For the simultaneous estimation we propose to minimize
	\begin{align}
		\label{obj:1}
		\frac{1}{2n}\sum_{k=1}^{K}||\mathds{X}^{(k)}(\mathbf{I}-\mathbf{\Theta}^{(k)})||_{F}^{2}+\lambda_{1}||\mathbf{H}||_{1}+\lambda_{2}\sum_{k=1}^{K}||\mathbf{\Gamma}^{(k)}||_{1}
	\end{align}
	over all $\mathbf{H}$ and $\mathbf{\Gamma}^{(k)}$ specified above, where $n=max\{n_{k}:k=1,\ldots,K\}$. The first penalty function penalizes the common factors $\eta_{lj}$ and is responsible for identifying the zeros across all coefficient vectors $\bm{\theta}_{j}^{(1)},\ldots,\bm{\theta}_{j}^{(K)}$. That is, if $\eta_{lj}$ is zero then vertex $l$ is not in the neighborhood of $j$ in all $K$ graphs. The second penalty function penalizes the individual factors $\gamma_{lj}^{(k)}$ and is responsible for identifying the zeros in $\bm{\theta}_{j}^{(k)}$ specific to each individual graph. That is, for a non-zero $\eta_{lj}$ some of the coefficients $\gamma_{lj}^{(1)},\ldots,\gamma_{lj}^{(K)}$ can be zero, which means that $l$ may be absent from the neighborhood of $j$ in some of the $K$ graphs but present in others.
	
	However, the objective function \eqref{obj:1} is difficult to minimize because of its complexity. It involves two groups of variables over which we have to optimize, and two parameters that we have to tune. As will be shown in the Theorem \ref{thm:2}, \eqref{obj:1} is equivalent to the much simpler form
		\begin{align}
			\label{obj:3}
			\frac{1}{2n}\sum_{k=1}^{K}||\mathds{X}^{(k)}(\mathbf{I}-\mathbf{\Theta}^{(k)})||_{F}^{2}+2(\lambda_{1}\lambda_{2})^{1/2}\sum_{l\ne j}\left(\sum_{k=1}^{K}|\theta_{lj}^{(k)}|\right)^{1/2}.
		\end{align}
		Let $\mathbf{\Theta}=(\mathbf{\Theta}^{(1)},\ldots,\mathbf{\Theta}^{(K)})$ and $\mathbf{\Gamma}=(\mathbf{\Gamma}^{(1)},\ldots,\mathbf{\Gamma}^{(K)})$. The proof of the following theorem can be found in the supplementary material.
		\medskip
		\begin{theorem}\label{thm:2}
			If $(\hat{\mathbf{H}},\hat{\mathbf{\Gamma}})$ is a local minimizer of \eqref{obj:1}, then there exists a local minimizer $\hat{\mathbf{\Theta}}$ of \eqref{obj:3} such that $\hat{\mathbf{\Theta}}^{(k)}=\hat{\mathbf{H}}\circ\hat{\mathbf{\Gamma}}^{(k)}$ for all $k$. Conversely, if $\hat{\mathbf{\Theta}}$ is a local minimizer of \eqref{obj:3}, then there exists a local minimizer $(\hat{\mathbf{H}},\hat{\mathbf{\Gamma}})$ of \eqref{obj:1} such that $\hat{\mathbf{H}}\circ\hat{\mathbf{\Gamma}}^{(k)}=\hat{\mathbf{\Theta}}^{(k)}$ for all $k$.
	\end{theorem}

	\section{Computation}
	\subsection{Penalty linearization}
	Writing $2(\lambda_{1}\lambda_{2})^{1/2}$ as $\lambda$, the objective function \eqref{obj:3} becomes
	\begin{align}\label{obj:4}
		\frac{1}{2n}\sum_{k=1}^{K}||\mathds{X}^{(k)}(\mathbf{I}-\mathbf{\Theta}^{(k)})||_{F}^{2}+\lambda\sum_{l\ne j}\left(\sum_{k=1}^{K}|\theta_{lj}^{(k)}|\right)^{1/2}.
	\end{align}
	Due to the presence of the square root in the penalty function, \eqref{obj:4} is not convex. To tackle this issue we approximate \eqref{obj:4} by using the Local Linear Approximation (LLA) method developed in \cite{zou2008one}, which proceeds as follows. Given an initial value $\hat{\mathbf{\Theta}}_{(0)}=\left(\hat{\mathbf{\Theta}}_{(0)}^{(1)},\ldots,\hat{\mathbf{\Theta}}_{(0)}^{(K)}\right)$ that is close to the true value, we locally approximate the penalty function by a linear function
	\begin{align*}
		\left(\sum_{k=1}^{K}|\theta_{lj}^{(k)}|\right)^{1/2}\approx \left(\sum_{k=1}^{K}|\hat{\theta}_{(0),lj}^{(k)}|\right)^{1/2}+\tau_{(0),lj}\left(\sum_{k=1}^{k}|\theta_{lj}^{(k)}|-\sum_{k=1}^{k}|\hat{\theta}_{(0),lj}^{(k)}|\right),
	\end{align*}
	where $\tau_{(0),lj}=2^{-1}\left(\sum_{k=1}^{K}|\hat{\theta}_{(0),lj}^{(k)}|\right)^{-1/2}$. Then, at the $t$-th iteration of the LLA algorithm, problem \eqref{obj:4} is decomposed into $K$ individual optimization problems
	\begin{align}\label{obj:5}
		\hat{\mathbf{\Theta}}_{(t)}^{(k)}=\argmin_{\mathbf{\Theta}\in\mathbb{R}^{p\times p},\diag(\mathbf{\Theta})=\mathbf{0}}\frac{1}{2n}||\mathds{X}^{(k)}(\mathbf{I}-\mathbf{\Theta})||_{F}^{2}+\lambda\sum_{l\ne j}\tau_{(t-1),lj}|\theta_{lj}|,
	\end{align}
	where $\tau_{(t-1),lj}$ is defined as above for $\hat{\mathbf{\Theta}}_{(t-1)}$ and $l\ne j$. In the asymptotics sections, it will be shown that if the initial estimator and the data satisfy certain conditions, then with only one iteration we can get an estimator with the oracle property.

	\subsection{ADMM algorithm for optimization}
	We employ the Alternating Direction Method of Multipliers \citep[ADMM;][]{boyd2011distributed} to solve the optimization problem \eqref{obj:5} for each subpopulation. Since this procedure is the same for all subpopulations $k$ and all LLA iterations $t$, we omit the subscript $(t)$ and the superscript $(k)$ in this section. The primal problem is given by
	\begin{equation}\label{primal1}
		\begin{split}
			&\text{minimize}\quad \frac{1}{2n}||\mathds{X}(\mathbf{I}-\mathbf{\Theta})||_{F}^{2}+\lambda\sum_{l\ne j}\tau_{lj}|\theta_{lj}|\\
			&\text{subject to}\quad\mathbf{\Theta}\in\mathbb{R}^{p\times p},\quad\diag(\mathbf{\Theta})=\mathbf{0}.
		\end{split}
	\end{equation}	
	To free ourselves of the zero diagonal constraint we reformulate the problem into an equivalent form. Let $\mathds{X}_{-j}$ denote the matrix $\mathds{X}$ with its $j$-th column removed, $\boldsymbol{\tau}_{j}$ denote the vector of weights $(\tau_{1j},\ldots,\tau_{pj})^{\intercal}$, and $\boldsymbol{\theta}_{j,-j},\boldsymbol{\tau}_{j,-j}$ the vectors $\boldsymbol{\theta}_{j},\boldsymbol{\tau}_{j}$ with their $j$-th elements removed. Define
	\begin{align*}
		\mathds{Y}=\left(
		\begin{array}{c}
			\mathds{X}_{1}\\
			\vdots\\
			\mathds{X}_{p}
		\end{array}
		\right),\quad
		\mathds{Z}=\left(
		\begin{array}{ccc}
			\mathds{X}_{-1}&\ldots&\mathbf{0}\\
			\vdots&\ddots&\vdots\\
			\mathbf{0}&\ldots&\mathds{X}_{-p}
		\end{array}
		\right),\quad
		\mathbf{v}=\left(
		\begin{array}{c}
			\boldsymbol{\theta}_{1,-1}\\
			\vdots\\
			\boldsymbol{\theta}_{p,-p}
		\end{array}
		\right),\quad
		\mathbf{w}=\left(
		\begin{array}{c}
			\boldsymbol{\tau}_{1,-1}\\
			\vdots\\
			\boldsymbol{\tau}_{p,-p}
		\end{array}
		\right).
	\end{align*}
	Then, the primal problem in \eqref{primal1} is equivalent to
	\begin{equation}\label{primal2}
		\begin{split}
			&\text{minimize}\quad\frac{1}{2n}||\mathds{Y}-\mathds{Z}\mathbf{v}||_{2}^{2}+\lambda\sum_{l=1}^{p(p-1)} w_{l}|v_{l}|\\
			&\text{subject to}\quad\mathbf{v}\in\mathbb{R}^{p(p-1)}.
		\end{split}
	\end{equation}	
	The dual problem of \eqref{primal2} is given by
	\begin{equation}\label{dual}
		\begin{split}
			&\text{minimize}\quad\frac{1}{2n}||\mathds{Y}-\mathds{Z}\mathbf{v}||_{2}^{2}+\lambda\sum_{l=1}^{p(p-1)} w_{l}|r_{l}|\\
			&\text{subject to}\quad\mathbf{v},\mathbf{r}\in\mathbb{R}^{p(p-1)},\quad\mathbf{v}-\mathbf{r}=\mathbf{0}.
		\end{split}
	\end{equation}
	The iteration formulas for the ADMM are given by
	\begin{equation}\label{admmjns}
		\begin{split}
			\mathbf{v}^{t+1}&=(\mathds{Z}^{\intercal}\mathds{Z}+nb\mathbf{I}_{p(p-1)})^{-1}[\mathds{Z}^{\intercal}\mathds{Y}+nb(\mathbf{r}^{t}-\mathbf{u}^{t})],\\
			\mathbf{r}^{t+1}&=\max\left(\mathbf{v}^{t+1}+\mathbf{u}^{t}-\frac{\lambda}{b}\mathbf{w},\mathbf{0}\right)-\max\left(-\mathbf{v}^{t+1}-\mathbf{u}^{t}-\frac{\lambda}{b}\mathbf{w},\mathbf{0}\right),\\
			\mathbf{u}^{t+1}&=\mathbf{u}^{t}+\mathbf{v}^{t+1}-\mathbf{r}^{t+1},
		\end{split}
	\end{equation}	
	where the $\max$ is taken for each component separately and $b$ is the step size. In practice, because of the size of $\mathbf{v}^{t+1}$, we compute it iteratively. This is possible as will be shown below.

	\subsection{Computational complexity}
	To show that our approach is computationally faster than the joint graphical lasso of \cite{guo2011joint} we first illustrate the ADMM for the method developed therein. The primal problem for the joint graphical lasso is
	\begin{equation*}
		\begin{split}
			&\text{minimize}\quad\tr(\mathbf{S}\mathbf{\Omega})-\log\det(\mathbf{\Omega})+\lambda\sum_{l<j}q_{lj}|\omega_{lj}|\\
			&\text{subject to}\quad\mathbf{\Omega}\succ\mathbf{0},
		\end{split}
	\end{equation*}
	and the dual is
	\begin{align*}
		&\text{minimize}\quad\tr(\mathbf{S}\mathbf{\Omega})-\log\det(\mathbf{\Omega})+\lambda\sum_{l<j}q_{lj}|z_{lj}|\\
		&\text{subject to}\quad\mathbf{\Omega}\succ\mathbf{0},\quad\mathbf{\Omega}-\mathbf{Z}=\mathbf{0}.
	\end{align*}
	The iteration formulas are given by
	\begin{equation}\label{admmjglasso}
		\begin{split}
			\mathbf{\Omega}^{t+1}&=\mathbf{Y}\left\{\frac{1}{2b}\left[\mathbf{\Lambda}+\left(\mathbf{\Lambda}^{2}+4b\mathbf{I}_{p}\right)^{1/2}\right]\right\}\mathbf{Y}^{\intercal},\\
			\mathbf{Z}^{t+1}&=\max\left(\mathbf{\Omega}^{t+1}+\mathbf{U}^{t}-\frac{\lambda}{b}\mathbf{Q},\mathbf{0}\right)-\max\left(-\mathbf{\Omega}^{t+1}-\mathbf{U}^{t}-\frac{\lambda}{b}\mathbf{Q},\mathbf{0}\right),\\
			\mathbf{U}^{t+1}&=\mathbf{U}^{t}+\mathbf{\Omega}^{t+1}-\mathbf{Z}^{t+1},
		\end{split}
	\end{equation}
	where the $\max$ is performed componentwise, $\mathbf{Q}=(q_{lj})$ is the symmetric matrix of weights, $\mathbf{Y}$ is a matrix whose columns are the eigenvectors, $\mathbf{\Lambda}$ is a diagonal matrix of the eigenvalues obtained by performing spectral decomposition on the symmetric matrix
	\begin{align*}
		b(\mathbf{Z}^{t}-\mathbf{U}^{t})-\mathbf{S}.
	\end{align*}
	Thus, the most time consuming part of \eqref{admmjglasso} is that of performing spectral decomposition on $b(\mathbf{Z}^{t}-\mathbf{U}^{t})-\mathbf{S}$, which is of computational complexity $\mathcal{O}(p^{3})$.
	
	On the other hand, the most time consuming part of \eqref{admmjns} is that of inverting the matrix $\mathds{Z}^{\intercal}\mathds{Z}+nb\mathbf{I}_{p(p-1)}$, which is of computational complexity $\mathcal{O}(p^{6})$. Because of its block-diagonal form, it is equivalent to inverting the matrices $\mathds{X}_{-j}^{\intercal}\mathds{X}_{-j}+nb\mathbf{I}_{p-1}$, $j=1,\ldots,p$, which can be done swiftly, as the next proposition shows. 
		\medskip
		\begin{proposition}\label{prop:1}
			Let $\mathbf{M}$ denote the matrix $\mathds{X}\mathds{X}^{\intercal}+nb\mathbf{I}_{n}$. The following identity holds.
			\begin{align*}
				(\mathds{X}_{-j}^{\intercal}\mathds{X}_{-j}+nb\mathds{I}_{p-1})^{-1}=\frac{1}{nb}\left(\mathbf{I}_{p-1}-\mathds{X}_{-j}^{\intercal}\mathbf{M}^{-1}\mathds{X}_{-j}+\frac{\mathds{X}_{-j}^{\intercal}\mathbf{M}^{-1}\mathds{X}_{j}\mathds{X}_{j}^{\intercal}\mathbf{M}^{-1}\mathds{X}_{-j}}{1-\mathds{X}_{j}^{\intercal}\mathbf{M}^{-1}\mathds{X}_{j}}\right).
			\end{align*}
		\end{proposition}
		Proposition \ref{prop:1} helps us reduce the computational complexity of \eqref{admmjns}, as is shown in the following proposition.
		\medskip
		\begin{proposition}\label{prop:2}
			The computational complexity of \eqref{admmjns} is $\mathcal{O}(np^{2})$.
	\end{proposition}

	\section{Asymptotics}
	In this section we prove the asymptotic consistency of the one-step version of \eqref{obj:5}, that is, the estimator produced by the first step of the LLA algorithm uncovers the true graph with probability tending to 1. Note that \eqref{obj:5} is essentially the summation of $p$ adaptive lasso problems \citep{zou2006adaptive}. In this section, since $k$ does not depend on $n$, we assume that we have the same number of samples $n$ for each subpopulation.
	
	The asymptotic consistency of the adaptive lasso has been established by \cite{huang2008adaptive}, where they proved that if the initial estimate satisfies certain conditions, then we choose the right variables with probability tending to 1. The difference of our case is that we are dealing with $p_{n}$ regressions each having $p_{n}-1$ variables. Thus, if we wish to prove graph consistency, we need to establish the asymptotic consistency uniformly for all regressions with diverging $p_{n}$.
	
	To do so, we begin by defining some useful notation. Let $\bm{\theta}_{j}^{(k)}$ denote the vector of regression coefficients defined in section \ref{methodology}. We assume that $\bm{\theta}_{j}^{(k)}$ has $q_{nj}^{(k)}$ nonzero elements and $s_{nj}^{(k)}$ zero elements, so that $q_{nj}^{(k)}+s_{nj}^{(k)}=p_{n}-1$. Without loss of generality, we assume the first $q_{nj}^{(k)}$ elements of $\bm{\theta}_{j}^{(k)}$ to be nonzero, and the next $s_{nj}^{(k)}$ elements to be 0. We denote the $l$-th element of $\bm{\theta}_{j}^{(k)}$ by $\theta_{lj}^{(k)}$ and define
	\begin{align*}
		b_{n}&=\min\left\{|\theta_{lj}^{(k)}|:j=1,
		\ldots,p_{n},\,l=1,\ldots,q_{nj}^{(k)}\right\}\\
		q_{n}&=\max\{q_{nj}^{(k)}:j=1,\ldots,p_{n}\}\\
		s_{n}&=\max\{s_{nj}^{(k)}:j=1,\ldots,p_{n}\}
	\end{align*}
	Notice that $b_{n},q_{n}$ and $s_{n}$ depend on $k$. However, because $k$ does not depend on $n$, for simplicity, we write $b_{n},q_{n},s_{n}$ instead of $b_{n}^{(k)},q_{n}^{(k)},s_{n}^{(k)}$.
	
	Without loss of generality, we assume that the matrices $\mathds{X}^{(1)},\ldots,\mathds{X}^{(K)}$ are centered and scaled, that is
	\begin{align}\label{centerscale}
		\sum_{i=1}^{n}x_{ij}^{(k)}=0\quad\text{and}\quad\frac{1}{n}\sum_{i=1}^{n}(x_{ij}^{(k)})^{2}=1,
	\end{align}
	for all $j,k$. Furthermore, for each $j,k$, we compartmentalize $\mathds{X}_{-j}^{(k)}$ into $\left[\mathds{X}_{-j;1}^{(k)},\mathds{X}_{-j;2}^{(k)}\right]$, where $\mathds{X}_{-j;1}^{(k)}\in\mathbb{R}^{n\times q_{nj}^{(k)}}$ and $\mathds{X}_{-j;2}^{(k)}\in\mathbb{R}^{n\times s_{nj}^{(k)}}$, and define $\hat{\mathbf{\Sigma}}_{jj}^{(k)}$ to be the matrix
	\begin{align*}
		\frac{1}{n}\mathds{X}_{-j;1}^{(k)\intercal}\mathds{X}_{-j;1}^{(k)}.
	\end{align*}
	We denote by $\upsilon_{nj}^{(k)}$ its smallest eigenvalue.
	\bigskip
	\begin{assumption}\label{ass:1}
		There exists a positive number $\xi$ such that $\upsilon_{nj}^{(k)}>\xi$ for all $n,j$.
	\end{assumption}
	\bigskip
	
	Let $\tilde{\bm{\theta}}_{j}^{(k)}$ be an initial estimate of $\bm{\theta}_{j}^{(k)}$ for all $j,k$. By construction, the weight $\tau_{lj}$ is equal to
	\begin{align*}
		\frac{1}{2}\left(\sum_{k=1}^{K}|\tilde{\theta}_{lj}^{(k)}|\right)^{-1/2},\quad l=1,\ldots,p_{n}-1.
	\end{align*}
	Let $\sgn:\mathbb{R}\rightarrow\mathbb{R}$ denote the sign function, such that
	\begin{align*}
		\sgn(t)=
		\begin{cases*}
			-1&,\text{if}\quad t<0\\
			0&,\text{if}\quad t=0\\
			1&,\text{if}\quad t>0
		\end{cases*}
		.
	\end{align*}
	Define $\mathbf{s}_{j}^{(k)}=\left(s_{1j}^{(k)},\ldots,s_{q_{nj}^{(k)}j}^{(k)}\right)^{\intercal}$, where
	\begin{align*}
		s_{lj}^{(k)}=\tau_{lj}\sgn\left(\theta_{lj}^{(k)}\right),\quad l=1,\ldots,q_{nj}^{(k)},
	\end{align*}
	for all regressions $j=1,\ldots,p_{n}$.
	
	To establish uniform consistency for all regressions $j=1,\ldots,p_{n}$, we need stronger assumptions on the initial estimators $\tilde{\bm{\theta}}_{j}^{(k)}$ than those made in \cite{huang2008adaptive}, particularly in the tail probability behaviors. We make the following assumptions.
	
	\bigskip
	\begin{assumption}\label{ass:2}
		For each $\tilde{\theta}_{lj}^{(k)}$ there exists a nonrandom constant $h_{lj}^{(k)}$ such that, for all $t\geq 0$,
		\begin{align*}
			P\left(\max_{\substack{1\leq j\leq p_{n} \\ 1\leq l\leq p_{n}-1}}\left|\frac{\sum_{k=1}^{K}|h_{lj}^{(k)}|}{\sum_{k=1}^{K}|\tilde{\theta}_{lj}^{(k)}|}-1\right|\geq t\right)&\leq\exp(-C t),\\
			P\left(\max_{\substack{1\leq j\leq p_{n} \\ 1\leq l\leq p_{n}-1}}\left|\sum_{k=1}^{K}|\tilde{\theta}_{lj}^{(k)}|-\sum_{k=1}^{K}|h_{lj}^{(k)}|\right|\geq t\right)&\leq\exp(-C t).
		\end{align*}
		Furthermore, the constants $h_{lj}^{(k)}$ satisfy
		\begin{align*}
			\max_{\substack{1\leq j\leq p_{n} \\ 1\leq l\leq p_{n}-1}}\frac{1}{\sum_{k=1}^{K}|h_{lj}^{(k)}|}\leq M_{1},\quad\max_{\substack{1\leq j\leq p_{n} \\ 1\leq l\leq p_{n}-1}}\sum_{k=1}^{K}|h_{lj}^{(k)}|\leq M_{2}.
		\end{align*}
	\end{assumption}
	\bigskip
	
	We denote the $l$-th column of $\mathds{X}_{-j}^{(k)}$ by $(\mathds{X}_{-j}^{(k)})_{l}$.
	
	\bigskip
	\begin{assumption}\label{ass:3}
		For each $t\geq 0$ it holds that
		\begin{align*}
			P\left(\max_{\substack{1\leq j\leq p_{n} \\ q_{nj}^{(k)}+1\leq l\leq p_{n}-1}}\left|(\mathds{X}_{-j}^{(k)})_{l}^{\intercal}\mathds{X}_{-j;1}^{(k)}\hat{\mathbf{\Sigma}}_{jj}^{(k)-1}\mathbf{s}_{j}\right|\geq t\right)&\leq \exp(-Ct^{2}).
		\end{align*}
	\end{assumption}
	\bigskip
	\begin{assumption}\label{ass:4}
			The sequences $p_{n},b_{n},q_{n},\lambda_{n}$ satisfy
			\begin{align*}
				\frac{\log p_{n}}{nb_{n}^{2}}+\frac{\lambda_{n}^{2}q_{n}\log p_{n}}{b_{n}^{2}}+\frac{\log p_{n}}{\lambda_{n}^{2}+n}\rightarrow 0
			\end{align*}
			as $n\rightarrow +\infty$.
	\end{assumption}
	\bigskip

	We are now ready to prove the main theorem of this section. We need to introduce the notion of sign equality between vectors, which is crucial for our proof. For any vector $\mathbf{v}=(v_{1},\ldots,v_{t})^{\intercal}$, we denote its sign vector by $\sgn(\mathbf{v})=(\sgn(x_{1}),\ldots,\sgn(x_{t}))^{\intercal}$. We say that two vectors $\bm{v},\bm{u}\in\mathbb{R}^{t}$ are equal in sign if $\sgn(\bm{v})=\sgn(\bm{u})$, and denote this by $\bm{v}=_{s}\bm{u}$. 
	
	\bigskip
	\begin{theorem}\label{thm:3}
		Suppose that Assumptions \ref{ass:1},\,\ref{ass:2},\,\ref{ass:3},\,\ref{ass:4} hold. Then,
		\begin{align*}
			P\left(\hat{\mathcal{E}}^{(k)}=\mathcal{E}^{(k)}\right)\rightarrow 1\quad\text{as}\quad n\rightarrow+\infty.
		\end{align*}
	\end{theorem}	
	The proof Theorem \ref{thm:3} can be found in the supplementary material.

	\section{Simulations}\label{simulations}	
	In this section we use simulation to compare the performance of the SNS with three existing methods: neighborhood selection \citep{meinshausen2006high} as applied to each subpopulation, which we refer to as the individual neighborhood selection (INS), joint graphical lasso method \citep[JGL;][]{guo2011joint}, and graphical lasso method \citep{yuan2007model} as applied to each subpopulation, which we refer to as the individual graphical lasso (IGL). We compare the methods both in ROC curve performance as well as CPU execution time of the optimization algorithm.
	
	To generate the data, we first construct the edge sets for all subpopulations, and then form the precision matrices. To form the edge sets $\mathcal{E}^{(1)},\ldots,\mathcal{E}^{(K)}$, we follow two steps:
	\begin{enumerate}
		\item Randomly choose a set of pairs $(j,l)\in\mathcal{V}\times\mathcal{V},j<l$, as a percentage $\boldsymbol{s}$ of the total number of edges $\binom{p}{2}$. This set constitutes the common edge set of the $K$ graphs and is denoted by $\mathcal{A}$.
		
		\item For each subpopulation $k$, randomly choose a set of pairs as a percentage $\rho$ of the number of common edges and denote this set by $\mathcal{B}^{(k)}$. The sets $\mathcal{B}^{(1)},\ldots,\mathcal{B}^{(K)}$ must satisfy
		\begin{align*}
			\bigcup_{k=1}^{K}\mathcal{B}^{(k)}\cap\mathcal{A}=\varnothing\quad\text{and}\quad\bigcap_{k=1}^{K}\mathcal{B}^{(k)}=\varnothing.
		\end{align*}
	\end{enumerate}	
	These sets are the individual edge structure of each subpopulation. Combining the above, we define $\mathcal{E}^{(k)}=\mathcal{A}\cup\mathcal{B}^{(k)},k=1,\ldots,K$.
	
	To form the precision matrices $\mathbf{\Omega}^{(1)},\ldots,\mathbf{\Omega}^{(K)}$, we follow three steps:
	\begin{enumerate}
		\item Generate $a_{jl}^{(k)}$ and $b_{jl}^{(k)}$ independently from $\mathcal{U}(0.5,1)$ and the Rademacher distribution, respectively, to form the matrices
		\begin{align*}
			\mathbf{A}^{(k)}=	\begin{cases}	
				a_{jl}^{(k)}b_{jl}^{(k)} &,(j,l)\in\mathcal{E}^{(k)}\\
				0 &,\text{otherwise}
			\end{cases}.
		\end{align*}
		
		\item To ensure symmetry, let
		\begin{align*}
			\mathbf{C}^{(K)}=\frac{\mathbf{A}^{(k)}+(\mathbf{A}^{(k)})^{\intercal}}{2}.
		\end{align*}
		
		\item Let $c_{jl}^{(k)}$ be the $(j,l)$-th element of $\mathbf{C}^{(k)}$. To ensure positive definiteness , we use Gershgorin's Circle Theorem \citep{bell1965gershgorin} to define the precision matrices $\mathbf{\Omega}^{(k)}=(\omega_{jl}^{(k)})$, with diagonal elements $\omega_{jj}^{(k)}=\sum_{q\ne j}|c_{jq}^{(k)}|+1$,
		and off-diagonal elements $\omega_{jl}^{(k)}=c_{jl}^{(k)}$.
	\end{enumerate}	
	With $\mathbf{\Omega}^{(1)},\ldots,\mathbf{\Omega}^{(K)}$ thus constructed, we are now ready to generate the observed data for each subpopulation. To do so we generate $\mathbf{X}_{1}^{(k)},\ldots,\mathbf{X}_{n}^{(k)}$ from a $\mathcal{N}_{p}(\mathbf{0},(\mathbf{\Omega}^{(k)})^{-1})$, where $\mathbf{X}_{i}^{(k)}=\left(X_{i1}^{(k)},\ldots,X_{ip}^{(k)}\right)^{\intercal}$.
	
	We compare the four methods on 12 different scenarios, which consist of the combinations of the dimensions $p=100,1000,2000,3000$ and proportions $\rho=0,0.5,1$. For each of the $p$'s mentioned above the common structure consists of $s=5\cdot 10^{-3}$, $5\cdot 10^{-4}$, $15\cdot 10^{-5}$, $75\cdot 10^{-6}$, as proportions of $\binom{p}{2}$, to achieve high levels of sparsity. In all settings the number of samples is $n=100$.
	
	For the ROC curves we plot the average true positive rate (ATPR) against the average false positive rate (AFPR), over a range of values of $\lambda$. Specifically,
	\begin{align*}
		\text{ATPR}(\lambda)&=\frac{1}{K}\sum_{k=1}^{K}\frac{\sum_{1 \leq j < l \leq p}\mathds{1} \bigg (\theta_{jl}^{(k)}\ne 0 \, , \,\hat{\theta}_{jl}^{(k)}(\lambda) \ne 0 \bigg)}{\sum_{1 \leq j < l \leq p}\mathds{1} \bigg(\theta_{jl}^{(k)}\ne 0 \bigg)}\\
		\text{AFPR}(\lambda)&=\frac{1}{K}\sum_{k=1}^{K}\frac{\sum_{1 \leq j < l \leq p}\mathds{1} \bigg (\theta_{jl}^{(k)}=0\,,\, \hat{\theta}_{jl}^{(k)}(\lambda) \ne 0 \bigg)}{\sum_{1 \leq j < l \leq p}\mathds{1} \bigg(\theta_{jl}^{(k)}=0 \bigg)},
	\end{align*}
	where $\mathds{1}$ is the indicator function and $\hat{\theta}_{jl}^{(k)}$ is the estimate of $\theta_{jl}^{(k)}$ using tuning parameter $\lambda$. The above quantities are calculated for 100 equally spaced values of $\lambda$, starting at $10^{-5}$ and ending at $1$. Each scenario is simulated 5 times, and the final ROC curve is the average of them. The curves are shown in Figure \ref{figure:1} and their respective areas-under-curve (AUC) in Table \ref{tabel:1}. We did not include the ROC results for IGL and JGL when $p=2000$ and 3000, because computing them would take more than 48 hours which is the wall time for the open queue in the PSU super-computer. Finally, in Table \ref{tabel:2}, we demonstrate the CPU times of SNS against JGL for one iteration of the ADMM algorithm mentioned in section 3. All the experiments were conducted on a 2.2 GHz Intel Xeon processor.

	\begin{figure}
		\centering
		\begin{tabular}{ccc}
			\includegraphics[scale=0.33 ]{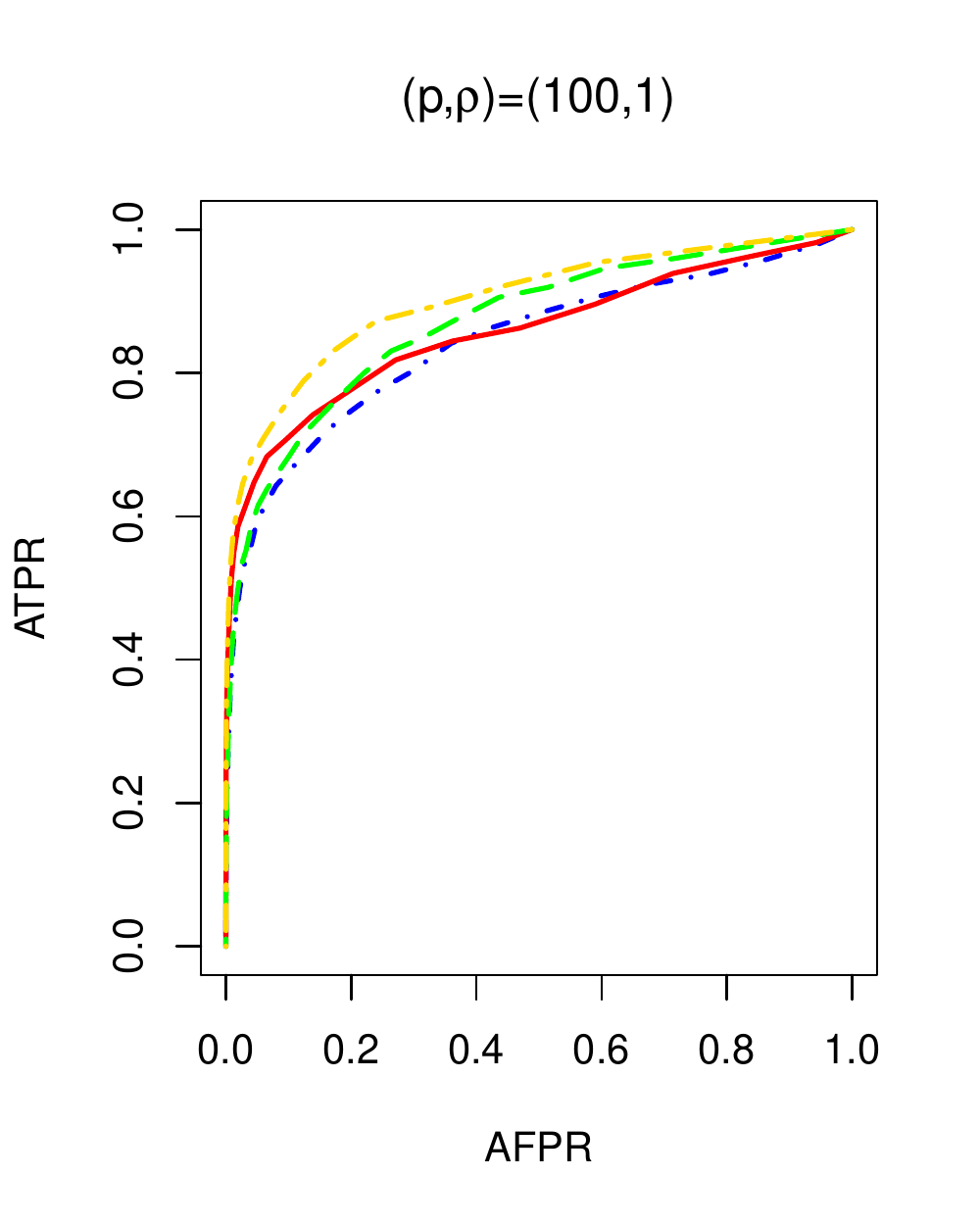} & \includegraphics[scale=0.33 ]{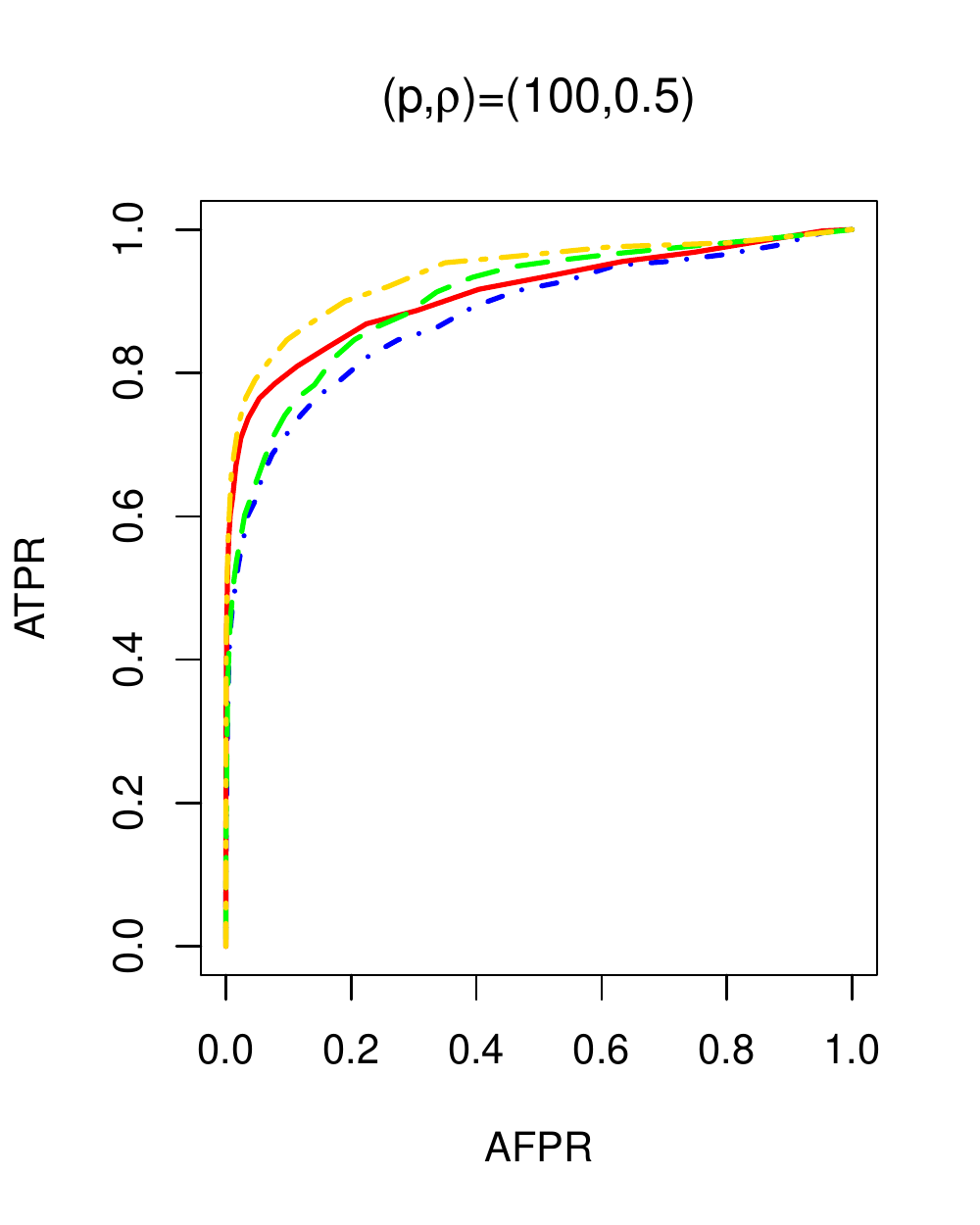} & 
			\includegraphics[scale=0.33 ]{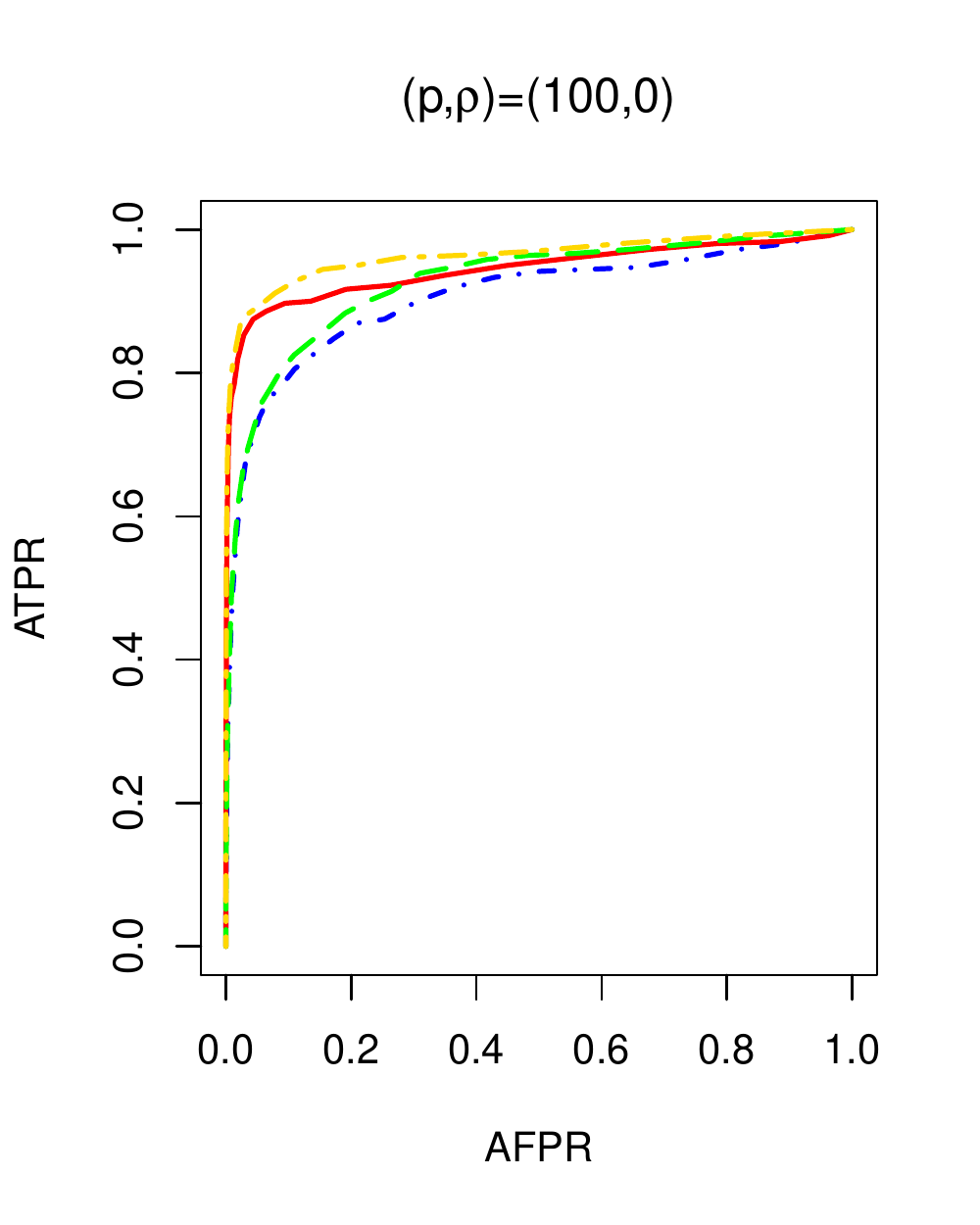}\\
			\includegraphics[scale=0.33 ]{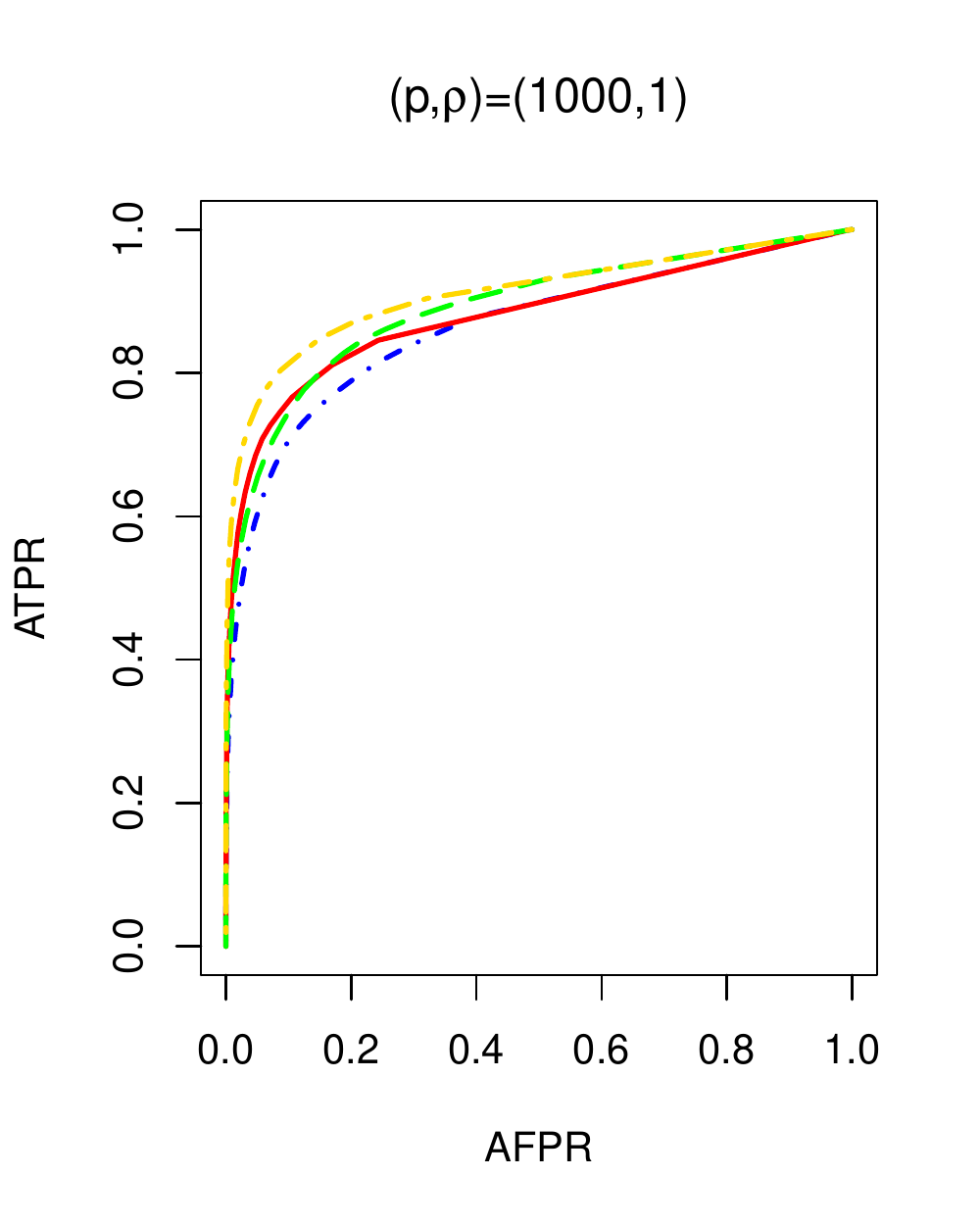} & \includegraphics[scale=0.33 ]{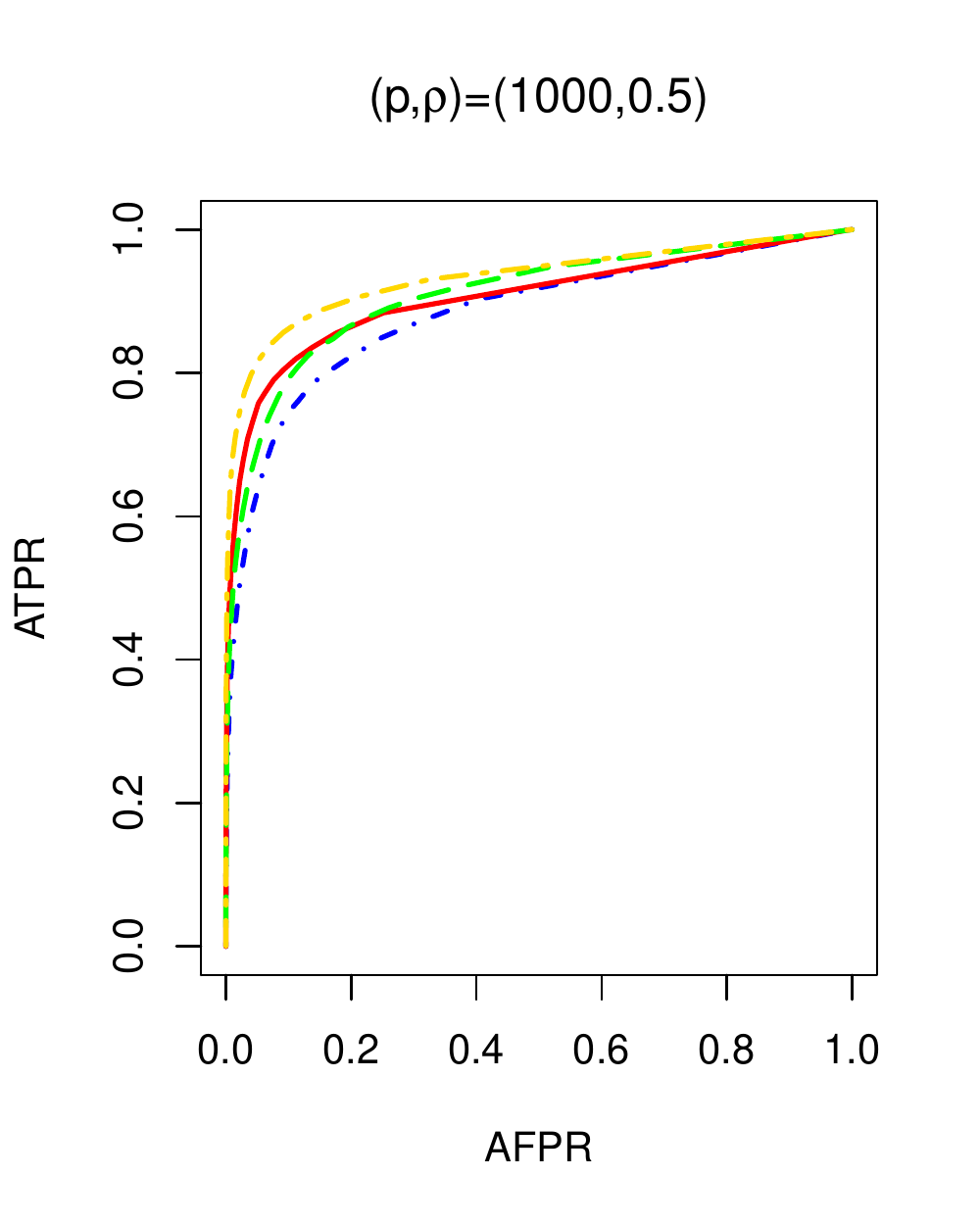} & 
			\includegraphics[scale=0.33 ]{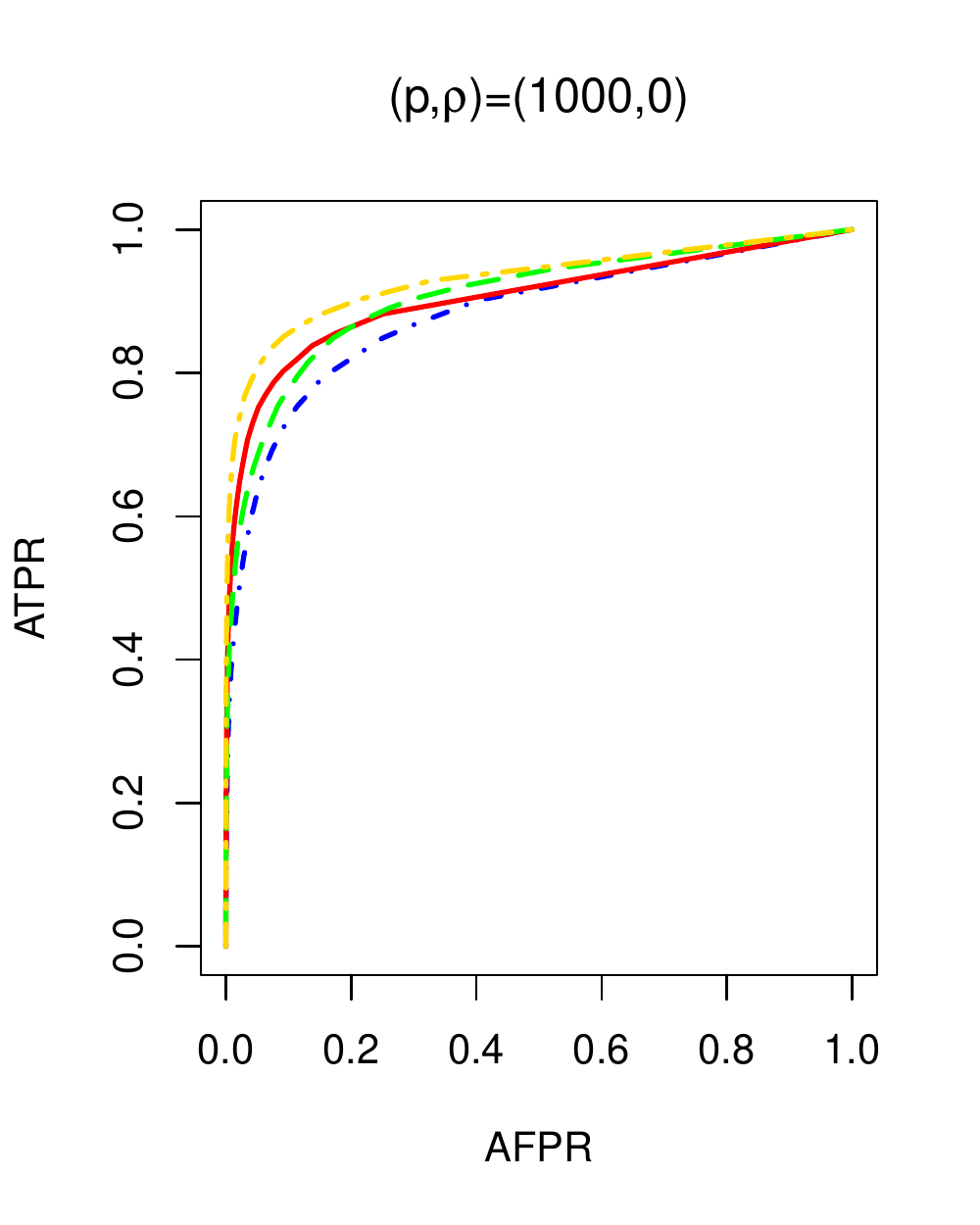}\\
			\includegraphics[scale=0.33 ]{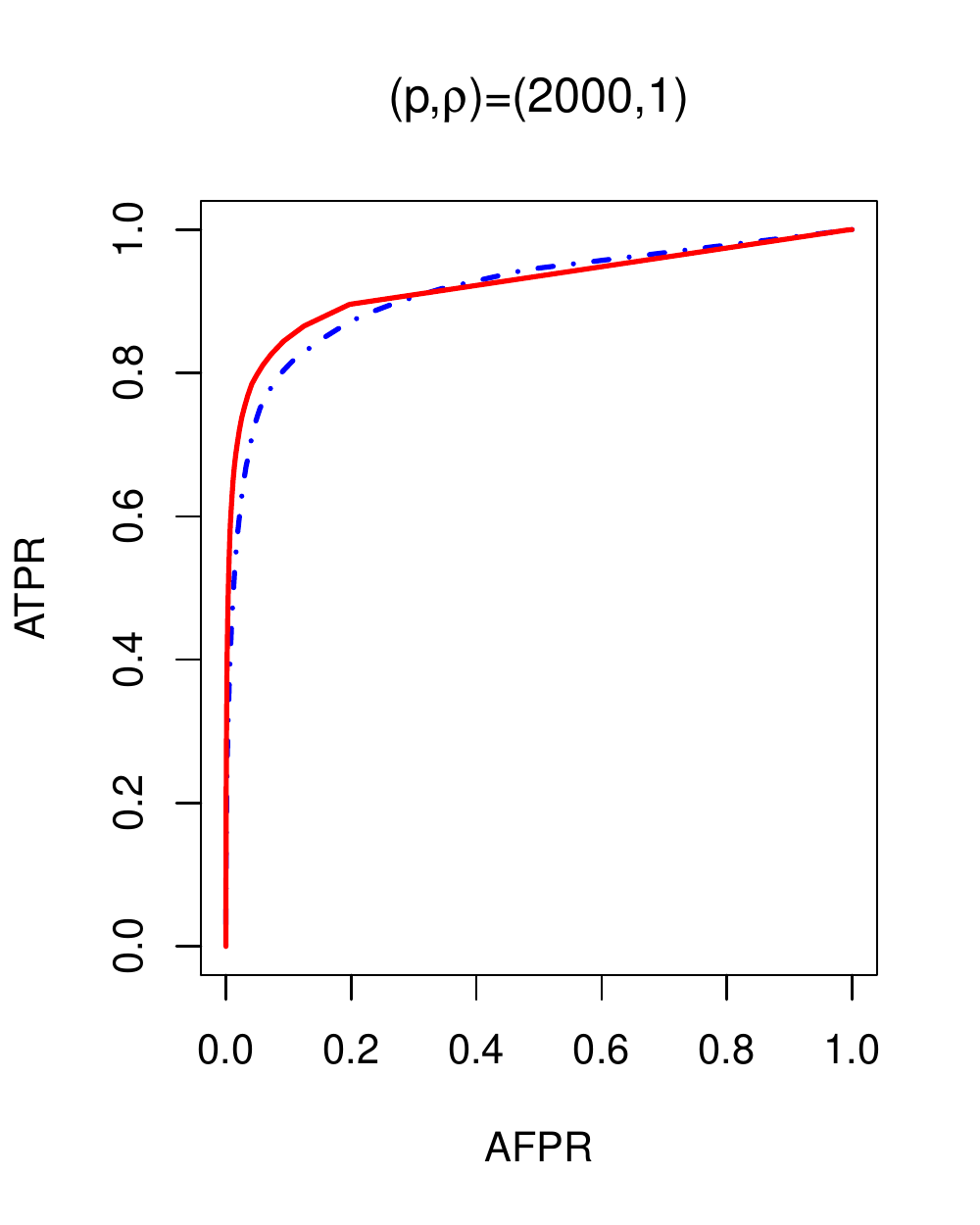} & \includegraphics[scale=0.33 ]{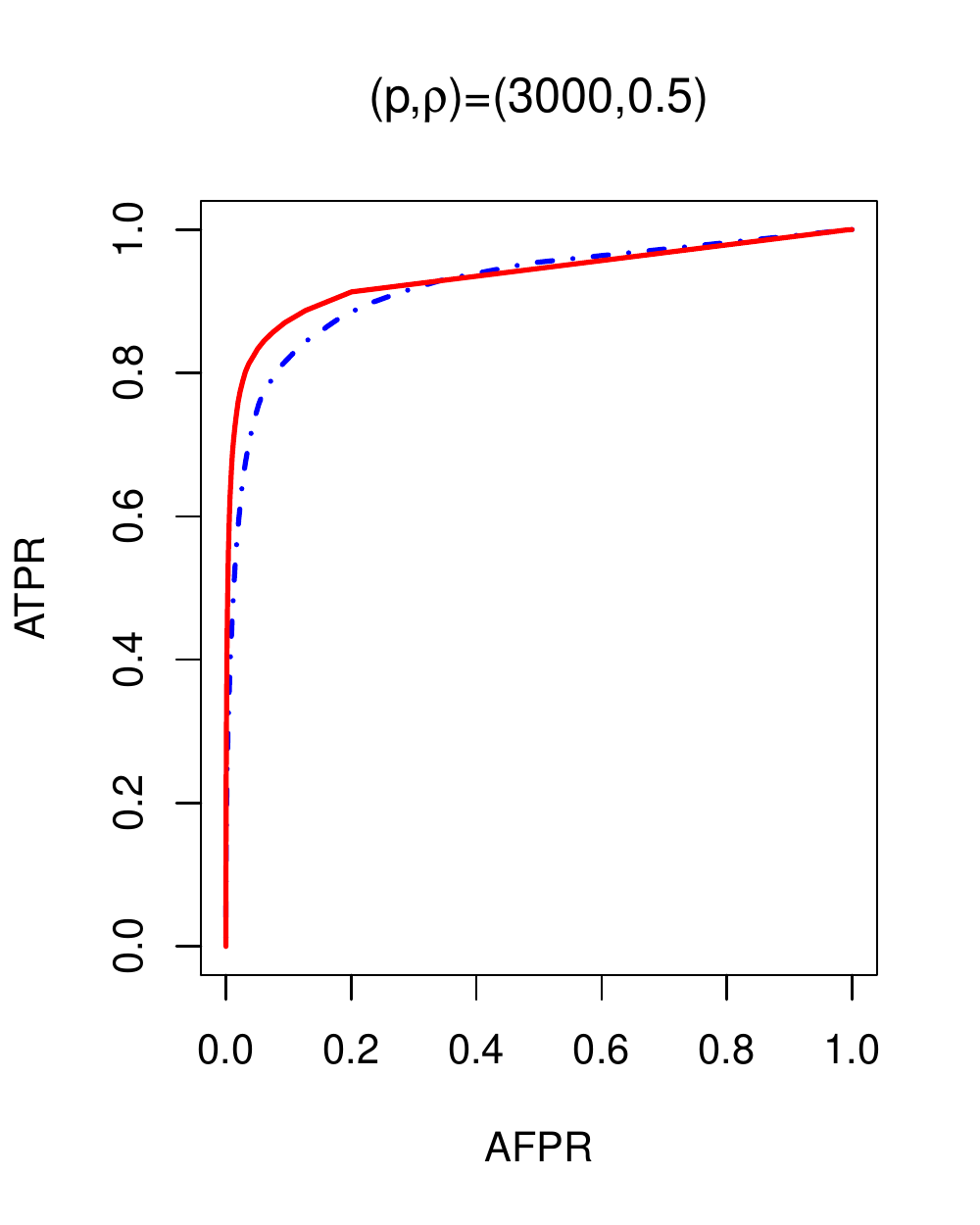} & 
			\includegraphics[scale=0.33 ]{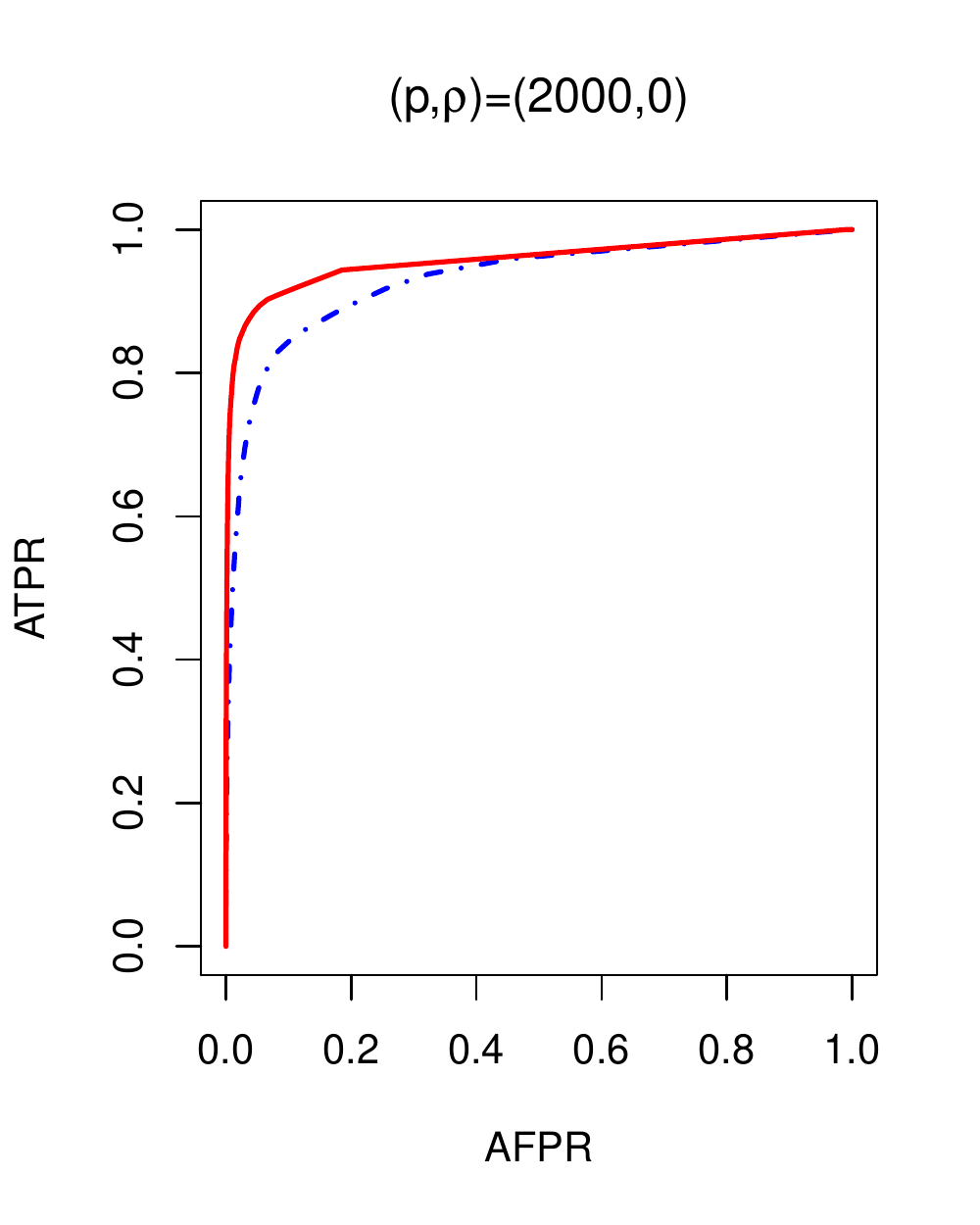}\\
			\includegraphics[scale=0.33 ]{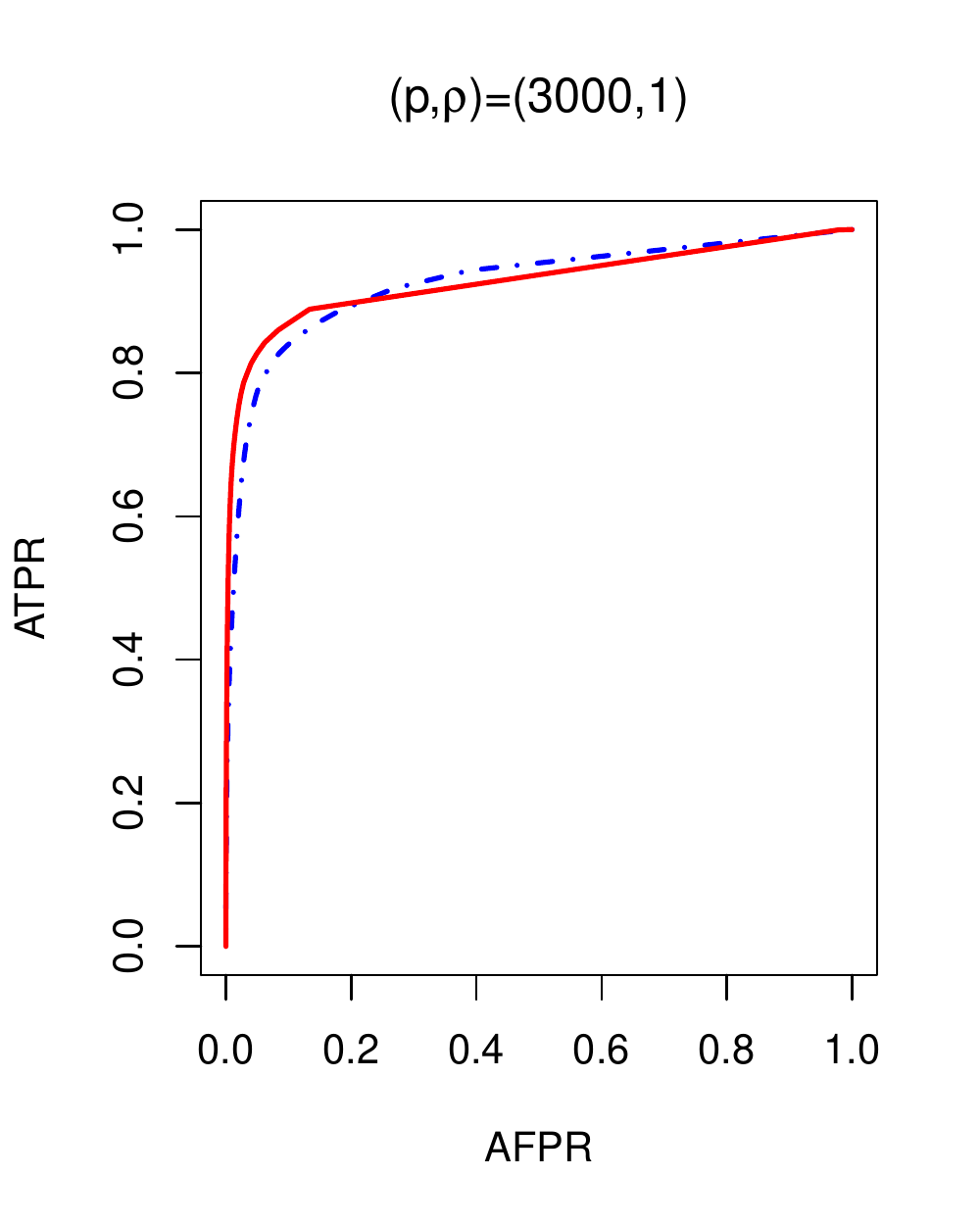} & \includegraphics[scale=0.33 ]{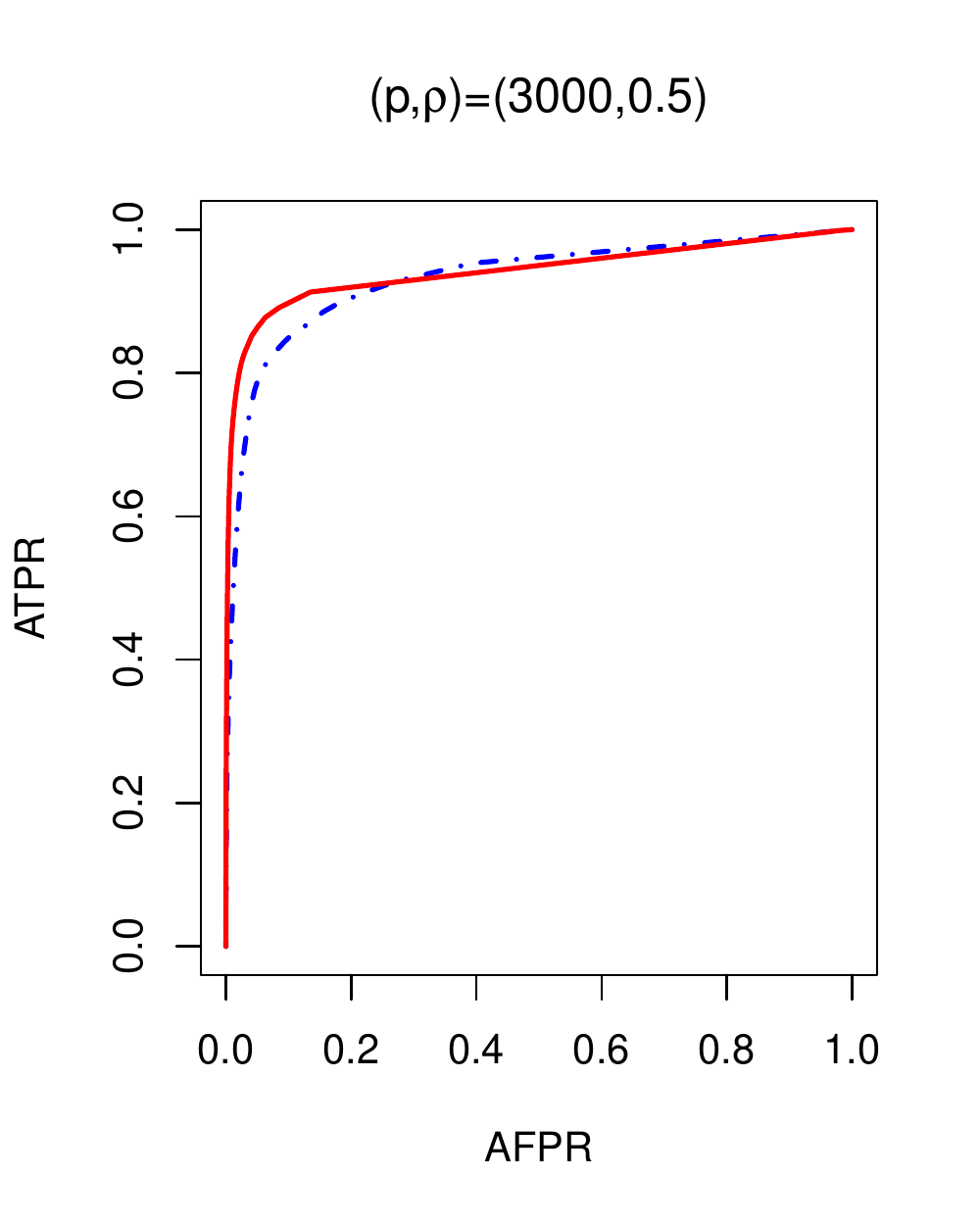} & \includegraphics[scale=0.33 ]{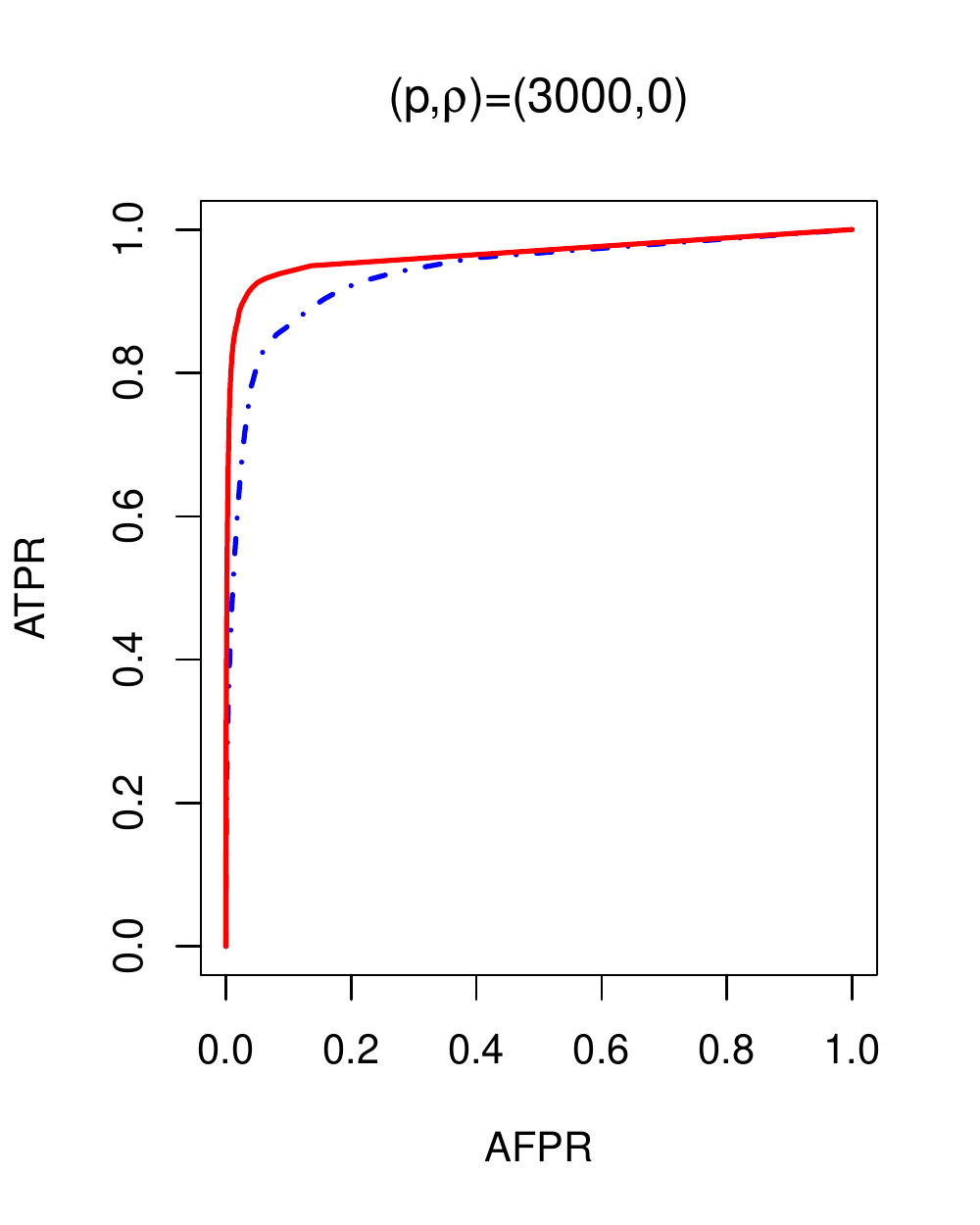}
		\end{tabular}
		\caption{ROC curves by SNS (red solid line), INS (blue dashed line), JGL (golden dashed line) and IGL (green dashed line) for different combinations of $(p,\rho)$. Note that for $p=2000,3000$, the computation of ROC results are only feasible for INS and SNS.}
		\label{figure:1}
	\end{figure}

	\begin{table}
		\centering
		\begin{tabular}{|M{2cm}|M{2cm}|M{2cm}|M{2cm}|M{2cm}|}
			\hline
			\multirow{2}{4em}{\centering$p$}   &\multirow{2}{4em}{\centering Method}&\multicolumn{3}{c|}{$\rho$}\\\cline{3-5}
			&&1 &0.5&0\\
			\hline
			\multirow{4}{4em}{\begin{center}100\end{center}}&SNS& 0.86  & 0.91 & 0.94 \\\cline{2-5}
			& INS& 0.84& 0.88& 0.91   \\\cline{2-5}
			& JGL& 0.90& 0.94& 0.96   \\\cline{2-5}
			& IGL& 0.87& 0.90& 0.93   \\
			\hline		
			\multirow{4}{4em}{\begin{center}1000\end{center}}&SNS& 0.88  & 0.89 & 0.90 \\\cline{2-5}
			& INS& 0.86& 0.87& 0.88   \\\cline{2-5}
			& JGL& 0.91& 0.92& 0.93   \\\cline{2-5}
			& IGL& 0.89& 0.90& 0.91   \\
			\hline		
			\multirow{2}{4em}{\centering2000}&SNS& 0.92  & 0.93 & 0.96 \\\cline{2-5}
			& INS& 0.91& 0.92 & 0.93   \\
			\hline		
			\multirow{2}{4em}{\centering3000}&SNS& 0.93  & 0.94 & 0.97 \\\cline{2-5}
			& INS& 0.92& 0.93 & 0.94   \\\cline{2-5}
			\hline
		\end{tabular}
		\caption{Table of the areas under the ROC curves from Figure \ref*{figure:1}. }
		\label{tabel:1}
	\end{table}

	\begin{table}
		\centering
		\begin{tabular}{|M{2cm}|M{2cm}|M{2cm}|M{2cm}|M{2cm}|}
			\hline
			\multirow{2}{4em}{\centering Method}&\multicolumn{4}{c|}{$p$}\\\cline{2-5}	
			& 100 & 1000 & 2000 & 3000\\
			\hline
			SNS&0.191 & 6.262 & 27.352 & 91.848\\ 
			\hline
			JGL&0.086 & 14.764 & 124.376 & 513.693\\
			\hline
		\end{tabular}
		\caption{CPU times for one iteration of the ADMM algorithm for the SNS and JGL for different numbers of features.}	
		\label{tabel:2}
	\end{table}
	
	To summarize the results, SNS performs significantly better than INS, due to the fact that the former takes advantage of the common structure. As expected, SNS does not perform as well as JGL. This is because the loss function of the SNS is a second-order approximation of the loss function of JGL. \cite{yuan2007model} discussed this point in the context of estimating a single graph. However, SNS requires significantly less computing time than JGL or IGL as reported in Table \ref{tabel:2}. Note that we did not test the computing time of IGL, because it has the same ADMM algorithm as JGL with the only difference that the weights are equal to 1. The performance of SNS and IGL is comparable but as with JGL it is infeasible to estimate the graph for large number of variables, such as the dataset we analyze next.

	\section{Application}	
	In this section we apply our method on the dataset mentioned in the introduction, which can be found in the GEO DataSet Browser under the name GDS2771. It consists of data from large airway epithelial cells. The data were collected from three subpopulations of smokers: with lung cancer, without lung cancer and with suspected lung cancer. The three subpopulations consist of 97, 90 and 5 subjects, respectively.
	
	We removed the subjects with suspected lung cancer, since our goal is not fitting a regression model to predict which of these subjects have cancer and which do not. We also removed all the duplicate variables along with variables that contained missing values. The remaining dataset consist of 14062 variables. Finally, we centered and scaled our data to match the input of the SNS.
	
	To get an initial estimate for $\mathbf{\Theta}$ in \eqref{obj:5} we applied neighborhood selection separately to each subpopulation. We set a $\lambda$ relatively small (0.05), so that we have a rich structure from which to select edges with our simultaneous estimation method. After getting the initial estimate, we applied our method to the dataset with a $\lambda=50$ to get a sparse graph. The first reason why $\lambda$ has to be this high is that the sample size (187) is very small compared to the size of the graph (14062), which requires strong regularization. The second reason is that we were only interested in a good visual representation of the graph of the dataset, which favored an interpretable sparse result.

	\begin{figure}
		\centering
		\begin{tabular}{c}
			\includegraphics[scale=0.4]{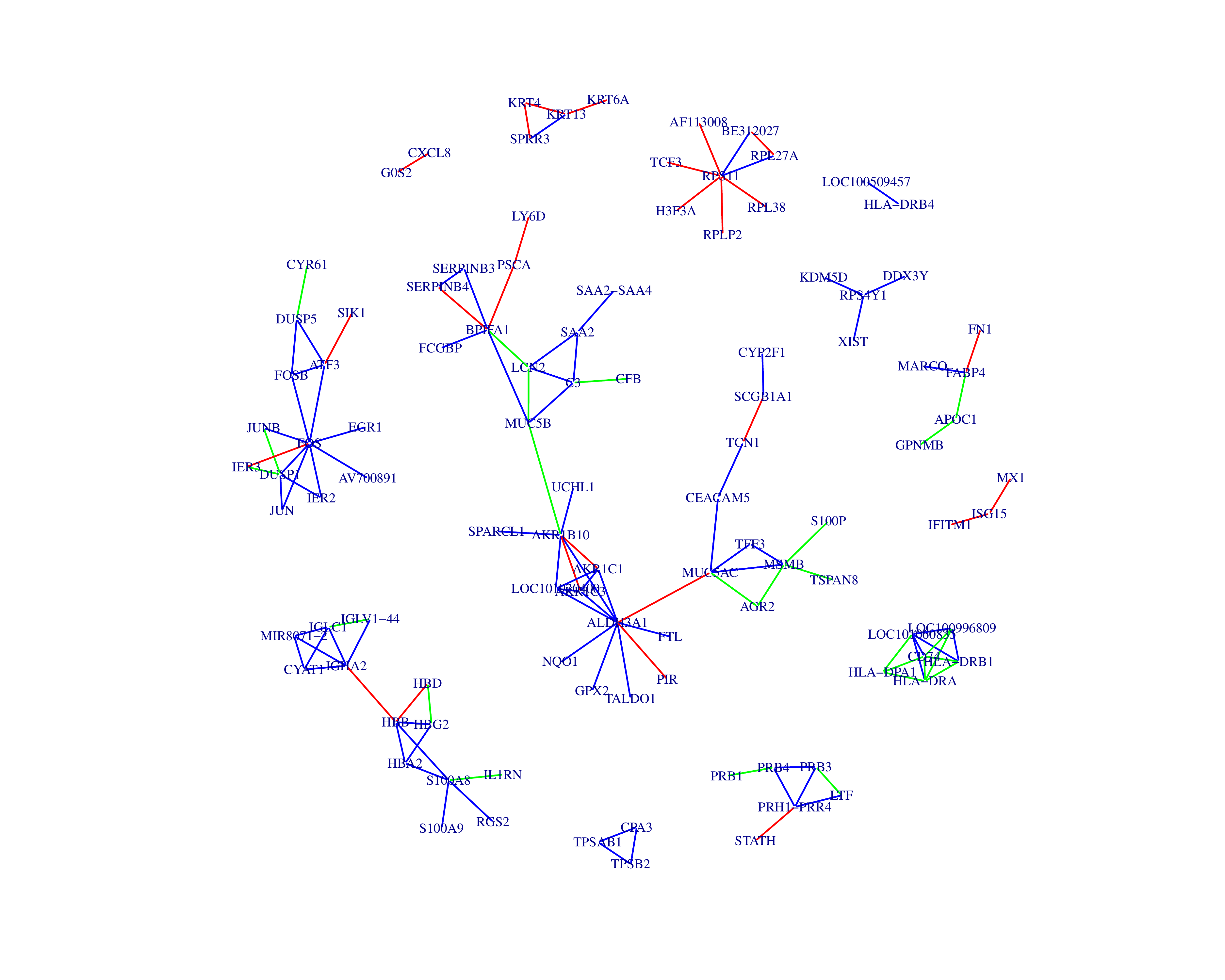}
		\end{tabular}
		\caption{Graph of lung cancer dataset, recovered with the simultaneous neighborhood selection method. Edges with blue color are common to both subpopulations, edges with green color are exclusive to the subpopulation with lung cancer and edges with red are exclusive to the subpopulation without lung cancer.}	
	\end{figure}

	\section{Appendix}
	Define the intermediate objective function:
		\begin{align}\label{obj:2}	
			\frac{1}{2n}\sum_{k=1}^{K}||\mathds{X}^{(k)}(\mathbf{I}-\mathbf{\Theta}^{(k)})||_{F}^{2}+||\mathbf{H}||_{1}+\lambda_{1}\lambda_{2}\sum_{k=1}^{K}||\mathbf{\Gamma}^{(k)}||_{1},
		\end{align}
		for $\mathbf{H}$ and $\mathbf{\Gamma}$ specified in subsection \ref{SNS}. We first show that the objective functions \eqref{obj:1} and \eqref{obj:2} are equivalent.
		\bigskip
		\begin{lemma}\label{lem:1}
			If $(\hat{\mathbf{H}},\hat{\mathbf{\Gamma}})$ is a local minimizer of \eqref{obj:1}, then there exists a local minimizer $(\tilde{\mathbf{H}},\tilde{\mathbf{\Gamma}})$ of \eqref{obj:2} such that $\tilde{\mathbf{H}}\circ\tilde{\mathbf{\Gamma}}^{(k)}=\hat{\mathbf{H}}\circ\hat{\mathbf{\Gamma}}^{(k)}$ for all $k$. Conversely, if $(\tilde{\mathbf{H}},\tilde{\mathbf{\Gamma}})$ is a local minimizer of \eqref{obj:2}, then there exists a local minimizer $(\hat{\mathbf{H}},\hat{\mathbf{\Gamma}})$ of \eqref{obj:1} such that $\hat{\mathbf{H}}\circ\hat{\mathbf{\Gamma}}^{(k)}=\tilde{\mathbf{H}}\circ\tilde{\mathbf{\Gamma}}^{(k)}$ for all $k$.
	\end{lemma}
	\begin{proof}	 	 
		Let $Q_{1}(\lambda_{1},\lambda_{2},\mathbf{H},\mathbf{\Gamma})$ and $Q_{2}(\lambda_{1}\lambda_{2},\mathbf{H},\mathbf{\Gamma})$ denote the objective functions \eqref{obj:1} and \eqref{obj:2}, respectively. Observe that
		\begin{align*}
			Q_{1}(\lambda_{1},\lambda_{2},\mathbf{H},\mathbf{\Gamma})&=Q_{2}(\lambda_{1}\lambda_{2},\lambda_{1}\mathbf{H},\lambda_{1}^{-1}\mathbf{\Gamma}),\\
			Q_{2}(\lambda_{1}\lambda_{2},\mathbf{H},\mathbf{\Gamma})&=Q_{1}(\lambda_{1},\lambda_{2},\lambda_{1}^{-1}\mathbf{H},\lambda_{1}\mathbf{\Gamma}).
		\end{align*}
		Since $(\hat{\mathbf{H}},\hat{\mathbf{\Gamma}})$ is a local minimizer of $Q_{1}(\lambda_{1},\lambda_{2},\cdot,\cdot)$, there exists $\delta>0$ such that for every $(\mathbf{H},\mathbf{\Gamma})$ with
		\begin{align*}
			||\mathbf{H}-\hat{\mathbf{H}}||_{1}+\sum_{k=1}^{K}||\mathbf{\Gamma}^{(k)}-\hat{\mathbf{\Gamma}}^{(k)}||_{1}<\delta,
		\end{align*}
		we have
		\begin{align*}
			Q_{1}(\lambda_{1},\lambda_{2},\hat{\mathbf{H}},\hat{\mathbf{\Gamma}})\leq Q_{1}(\lambda_{1},\lambda_{2},\mathbf{H},\mathbf{\Gamma}).
		\end{align*}
		Let $0<\delta^{*}\leq \delta\min(\lambda_{1},\lambda_{1}^{-1})$, and define $(\tilde{\mathbf{H}},\tilde{\mathbf{\Gamma}})=(\lambda_{1}\hat{\mathbf{H}},\lambda_{1}^{-1}\hat{\mathbf{\Gamma}})$. Then, for any $(\mathbf{H},\mathbf{\Gamma})$ satisfying
		\begin{align*}
			||\mathbf{H}-\tilde{\mathbf{H}}||_{1}+\sum_{k=1}^{K}||\mathbf{\Gamma}^{(k)}-\tilde{\mathbf{\Gamma}}^{(k)}||_{1}<\delta^{*},
		\end{align*}
		we have
		\begin{align*}
			||\lambda_{1}^{-1}\mathbf{H}-\hat{\mathbf{H}}||_{1}+\sum_{k=1}^{K}||\lambda_{1}\mathbf{\Gamma}^{(k)}-\hat{\mathbf{\Gamma}}^{(k)}||_{1}\leq \frac{||\mathbf{H}-\tilde{\mathbf{H}}||_{1}+\sum_{k=1}^{K}||\mathbf{\Gamma}^{(k)}-\tilde{\mathbf{\Gamma}}^{(k)}||_{1}}{\min(\lambda_{1},\lambda_{1}^{-1})}\leq \delta.
		\end{align*}
		Thus
		\begin{align*}
			Q_{1}(\lambda_{1},\lambda_{2},\hat{\mathbf{H}},\hat{\mathbf{\Gamma}})\leq Q_{1}(\lambda_{1},\lambda_{2},\lambda_{1}^{-1}\mathbf{H},\lambda_{1}\mathbf{\Gamma})\Rightarrow Q_{2}(\lambda_{1}\lambda_{2},\tilde{\mathbf{H}},\tilde{\mathbf{\Gamma}})\leq Q_{2}(\lambda_{1}\lambda_{2},\mathbf{H},\mathbf{\Gamma}),
		\end{align*}
		which means that $(\tilde{\mathbf{H}},\tilde{\mathbf{\Gamma}})$ is a local minimizer of \eqref{obj:2}. The other direction is proven similarly.
	\end{proof}
	
	\bigskip
	\setcounter{lemma}{1}
	\begin{lemma}\label{lem:2}
		Suppose $(\hat{\mathbf{H}},\hat{\mathbf{\Gamma}})$ is a local minimizer of \eqref{obj:2} and $\hat{\mathbf{\Theta}}^{(k)}=\hat{\mathbf{H}}\circ\hat{\mathbf{\Gamma}}^{(k)}$ for all $k$. Then for any $l,j$ the following statements hold true:
		\begin{enumerate}
			\item $\hat{\eta}_{lj}=0$ if and only if $\hat{\gamma}_{lj}^{(k)}=0$ for all $k=1,\ldots,K$;
			\item If $\hat{\eta}_{lj}\ne 0$, then $\hat{\eta}_{lj}=\left(\lambda_{1}\lambda_{2}\sum_{k=1}^{K}|\hat{\theta}_{lj}(k)|\right)^{1/2}$.
		\end{enumerate}
	\end{lemma}
	
	\begin{proof}
		1. If $l=j$, by definition $\gamma_{jj}^{(k)}=\eta_{jj}=0$ for all $k$. Suppose then $l\ne j$ and $\eta_{lj}=0$, then $\gamma_{lj}^{(1)},\ldots,\gamma_{lj}^{(K)}$ only appear in the third term in \eqref{obj:2}. Thus, in order to minimize $Q_{2}(\lambda_{1}\lambda_{2},\cdot,\cdot)$, we need $\gamma_{lj}^{(k)}=0$ for all $k=1,\ldots,K$. The reverse implication can be proved similarly.
		
		2. Suppose $\hat{\eta}_{lj}\ne 0$ and let
		\begin{align*}
			c=\frac{\left(\lambda_{1}\lambda_{2}\sum_{k=1}^{K}|\hat{\theta}_{lj}(k)|\right)^{1/2}}{\hat{\eta}_{lj}}.
		\end{align*}
		We will show $c=1$. By definition
		\begin{align*}
			\hat{\gamma}_{lj}^{(k)}=c\frac{\hat{\theta}_{lj}(k)}{\left(\lambda_{1}\lambda_{2}\sum_{k=1}^{K}|\hat{\theta}_{lj}(k)|\right)^{1/2}}.
		\end{align*}
		Suppose $c>1$. Since $(\hat{\mathbf{H}},\hat{\mathbf{\Gamma}})$ is a local minimizer of $Q_{2}(\lambda_{1}\lambda_{2},\cdot,\cdot)$, there exists $\delta>0$ such that for all $(\mathbf{H},\mathbf{\Gamma})$ with
		\begin{align*}
			||\mathbf{H}-\hat{\mathbf{H}}||_{1}+\sum_{k=1}^{K}||\mathbf{\Gamma}^{(k)}-\hat{\mathbf{\Gamma}}^{(k)}||_{1}<\delta,
		\end{align*}
		we have $Q_{2}(\lambda_{1}\lambda_{2},\hat{\mathbf{H}},\hat{\mathbf{\Gamma}})\leq Q_{2}(\lambda_{1}\lambda_{2},\mathbf{H},\mathbf{\Gamma})$. Furthermore, there exists $\delta^{*}\in(1,c)$ slightly greater that 1, such that for $(\tilde{\mathbf{H}},\tilde{\mathbf{\Gamma}})$ defined by
		\begin{align*}
			\begin{cases*}
				\tilde{\eta}_{l'j'}=\hat{\eta}_{l'j'}\text{ and }\tilde{\gamma}_{l'j'}^{(k)}=\hat{\gamma}_{l'j'}^{(k)}, &$(l',j')\ne (l,j)$\\
				\tilde{\eta}_{lj}=\delta^{*}\hat{\eta}_{lj}\text{ and }\tilde{\gamma}_{lj}^{(k)}=\frac{1}{\delta^{*}}\hat{\gamma}_{lj}^{(k)}
			\end{cases*}
		\end{align*}
		we have
		\begin{align*}
			||\tilde{\mathbf{H}}-\hat{\mathbf{H}}||_{1}+\sum_{k=1}^{K}||\tilde{\mathbf{\Gamma}}^{(k)}-\hat{\mathbf{\Gamma}}^{(k)}||_{1}<\delta.
		\end{align*}
		But this implies
		\begin{align*}
			Q_{2}(\lambda_{1}\lambda_{2},\hat{\mathbf{H}},\hat{\mathbf{\Gamma}})-Q_{2}(\lambda_{1}\lambda_{2},\tilde{\mathbf{H}},\tilde{\mathbf{\Gamma}})&=(1-\delta^{*})\hat{\eta}_{lj}+\left(1-\frac{1}{\delta^{*}}\right)\lambda_{1}\lambda_{2}\sum_{k=1}^{K}|\hat{\gamma}_{lj}^{(k)}|\\
			&=\frac{1}{c}(\delta^{*}-1)\left(\frac{c^{2}}{\delta^{*}}-1\right)\left(\lambda_{1}\lambda_{2}\sum_{k=1}^{K}|\hat{\theta}_{lj}(k)|\right)^{1/2}>0,
		\end{align*}
		which is impossible because $(\hat{\mathbf{H}},\hat{\mathbf{\Gamma}})$ is a local minimizer. Hence $c\leq 1$. Following the same argument we can show $c\geq 1$. Thus $c=1$.
	\end{proof}

	\begin{lemma}\label{lem:3}
		If $\{a_{l}^{(k)}:l\in\mathcal{V}\}$ and $\{b_{l}^{(k)}:l\in\mathcal{V}\}$ satisfy
		\begin{align*}
			0<|b_{l}^{(k)}|<\frac{1}{K}\min\left\{|a_{l}^{(k)}|:a_{l}^{(k)}\ne 0,l\in\mathcal{V},k=1,\ldots,K\right\}
		\end{align*}
		for all $l,k$, then, for each $l\in\mathcal{V}$,
		\begin{align*}
			\left|\left(\sum_{k=1}^{K}|a_{l}^{(k)}+b_{l}^{(k)}|\right)^{1/2}-\left(\sum_{k=1}^{K}|b_{l}^{(k)}|\right)^{1/2}\right|\leq \frac{1}{2}\frac{\sum_{k=1}^{K}|b_{l}^{(k)}|}{\left(\sum_{k=1}^{K}|a_{l}^{(k)}|-\sum_{k=1}^{K}|b_{l}^{(k)}|\right)^{1/2}}.
		\end{align*}
	\end{lemma}
	\begin{proof}
		First note that for any $0<x<y$, we have
		\begin{align}\label{ineq:1}
			|\sqrt{y+x}-\sqrt{y}|\leq |\sqrt{y-x}-\sqrt{y}|\leq \frac{1}{2}\frac{x}{\sqrt{y-x}}.
		\end{align}
		By the triangular inequality, 
		\begin{align*}
			\sum_{k=1}^{K}|a_{l}^{(k)}|-\sum_{k=1}^{K}|b_{l}^{(k)}|\leq \sum_{k=1}^{K}|a_{l}^{(k)}+b_{l}^{(k)}|\leq \sum_{k=1}^{K}|a_{l}^{(k)}|+\sum_{k=1}^{K}|b_{l}^{(k)}|.
		\end{align*}
		Therefore, for each $l\in\mathcal{V}$,
		\begin{align*}
			\left|\left(\sum_{k=1}^{K}|a_{l}^{(k)}+b_{l}^{(k)}|\right)^{1/2}-\left(\sum_{k=1}^{K}|a_{l}^{(k)}|\right)^{1/2}\right|
		\end{align*}
		is less than or equal to the maximum of
		\begin{align*}
			\left|\left(\sum_{k=1}^{K}|a_{l}^{(k)}|-\sum_{k=1}^{K}|b_{l}^{(k)}|\right)^{1/2}-\left(\sum_{k=1}^{K}|a_{l}^{(k)}|\right)^{1/2}\right|,\\
		\end{align*}
		and
		\begin{align*}
			\left|\left(\sum_{k=1}^{K}|a_{l}^{(k)}|+\sum_{k=1}^{K}|b_{l}^{(k)}|\right)^{1/2}-\left(\sum_{k=1}^{K}|a_{l}^{(k)}|\right)^{1/2}\right|.
		\end{align*}
		According to \eqref{ineq:1}, the maximum is less than or equal to
		\begin{align*}
			\frac{1}{2}\frac{\sum_{k=1}^{K}|b_{l}^{(k)}|}{\left(\sum_{k=1}^{K}|a_{l}^{(k)}|-\sum_{k=1}^{K}|b_{l}^{(k)}|\right)^{1/2}},
		\end{align*}
		as desired.
	\end{proof}
	
	We are now ready to establish the equivalence between the objective functions \eqref{obj:1} and \eqref{obj:3}, for which it suffices to show the equivalence of \eqref{obj:3} and \eqref{obj:2}, as we do next.\\

	\begin{large}
		\noindent\textbf{Proof of Theorem \ref{thm:2}}
	\end{large}
	
	\begin{proof}
		Let $Q_{3}(\lambda_{1}\lambda_{2},\mathbf{\Theta})$ denote the objective function \eqref{obj:3}, and let also $\lambda=\lambda_{1}\lambda_{2}$. Suppose $(\hat{\mathbf{H}},\hat{\mathbf{\Gamma}})$ is a local minimizer of \eqref{obj:2}. Then, there exists $\delta>0$ such that for all $(\mathbf{H},\mathbf{\Gamma})$ with
		\begin{align*}
			||\mathbf{H}-\hat{\mathbf{H}}||_{1}+\sum_{k=1}^{K}||\mathbf{\Gamma}^{(k)}-\hat{\mathbf{\Gamma}}^{(k)}||_{1}<\delta,
		\end{align*}
		we have
		\begin{align*}
			Q_{2}(\lambda,\hat{\mathbf{H}},\hat{\mathbf{\Gamma}})\leq Q_{2}(\lambda,\mathbf{H},\mathbf{\Gamma}).
		\end{align*}
		Let $\hat{\mathbf{\Theta}}$ be the estimator associated with $(\hat{\mathbf{H}},\hat{\mathbf{\Gamma}})$, that is, $\hat{\mathbf{\Theta}}^{(k)}=\hat{\mathbf{H}}\circ\hat{\mathbf{\Gamma}}^{(k)}$ for all $k$. Before proceeding further, we need to define some constants. Let
		\begin{align*}
			a&=\frac{1}{K}\min\left\{|\hat{\theta}_{lj}(k)|:\hat{\theta}_{lj}(k)\ne 0,l,j\in\mathcal{V},k=1,\ldots,K\right\},\\
			b&=\max\left\{|\hat{\theta}_{lj}(k)|:l,j\in\mathcal{V},k=1,\ldots,K\right\},\\
			c&=2p^{2}\max\left\{\lambda^{1/2},\lambda^{-1/2},\left(\frac{\lambda}{2Ka}\right)^{1/2},\left(\frac{2}{\lambda K a}\right)^{1/2}+\left(\frac{b}{\lambda K^{2}a^{2}}\right)^{1/2}\right\}.
		\end{align*}	
		Let also $0<\delta^{*}\leq \min\left(\frac{Ka}{2},a,1,c^{-2}\delta^{2}\right)$ and $\mathbf{D}=\left(\mathbf{D}^{(1)},\ldots,\mathbf{D}^{(K)}\right)$, where $\mathbf{D}^{(k)}=(d_{lj}^{(k)})$	
		and satisfies	
		$0<|d_{l}^{(k)}|$ for all $l,k$ and $\sum_{k=1}^{K}||\mathbf{D}^{(k)}||_{1}<\delta^{*}$. Let $\tilde{\mathbf{\Theta}}=\hat{\mathbf{\Theta}}+\mathbf{D}$. Then
		\begin{align*}
			\sum_{k=1}^{K}||\tilde{\mathbf{\Theta}}^{(k)}-\hat{\mathbf{\Theta}}^{(k)}||_{1}<\delta^{*}.
		\end{align*}
		This means that $\tilde{\mathbf{\Theta}}$ is a generic element of the ball with radius $\delta^{*}$ and center $\hat{\mathbf{\Theta}}$.
		
		Define $(\tilde{\mathbf{H}},\tilde{\mathbf{\Gamma}})$ by
		\begin{align*}
			\tilde{\eta}_{lj}=\left(\lambda\sum_{k=1}^{K}|\hat{\theta}_{lj}+d_{lj}^{(k)}|\right)^{1/2}\quad\text{and}\quad\tilde{\gamma}_{lj}^{(k)}=\frac{\hat{\theta}_{lj}(k)+d_{lj}^{(k)}}{\left(\lambda\sum_{k=1}^{K}|\hat{\theta}_{lj}(k)+d_{lj}^{(k)}|\right)^{1/2}}
		\end{align*}
		for all $l,j,k$. By Lemma \ref{lem:2},
		\begin{align*}
			Q_{2}(\lambda,\hat{\mathbf{H}},\hat{\mathbf{\Gamma}})=Q_{3}(\lambda,\hat{\mathbf{\Theta}}),
		\end{align*}
		and by the definition of $(\tilde{\mathbf{H}},\tilde{\mathbf{\Gamma}})$,
		\begin{align*}
			Q_{2}(\lambda,\tilde{\mathbf{H}},\tilde{\mathbf{\Gamma}})=Q_{3}(\lambda,\tilde{\mathbf{\Theta}}).
		\end{align*}
		If $\hat{\eta}_{lj}=0$, then $\hat{\gamma}_{lj}^{(k)}=\hat{\theta}_{lj}(k)=0$ for all $k$. Hence
		\begin{align*}
			|\tilde{\eta}_{lj}-\hat{\eta}_{lj}|=\left(\lambda\sum_{k=1}^{K}|d_{lj}^{(k)}|\right)^{1/2}\leq \left(\lambda\sum_{k=1}^{K}||\mathbf{D}^{(k)}||_{1}\right)^{1/2}<\lambda^{1/2}\delta^{*1/2}
		\end{align*}
		and
		\begin{align*}
			\sum_{k=1}^{K}|\tilde{\gamma}_{lj}^{(k)}-\hat{\gamma}_{lj}^{(k)}|=\frac{\sum_{k=1}^{K}|d_{lj}^{(k)}|}{\left(\lambda\sum_{k=1}^{K}|d_{lj}^{(k)}|\right)^{1/2}}\leq\lambda^{-1/2}\left(\sum_{k=1}^{K}||\mathbf{D}^{(k)}||_{1}\right)^{1/2}<\lambda^{-1/2}\delta^{*1/2}
		\end{align*}
		If $\hat{\eta}_{lj}\ne 0$, then by Lemma \ref{lem:3},
		\begin{align*}
			|\tilde{\eta}_{lj}-\hat{\eta}_{lj}|&\leq\frac{\lambda^{1/2}}{2}\frac{\sum_{k=1}^{K}|d_{lj}^{(k)}|}{\left(\sum_{k=1}^{K}|\hat{\theta}_{lj}(k)|-\sum_{k=1}^{K}|d_{lj}^{(k)}|\right)^{1/2}}\\
			&\leq\frac{\lambda^{1/2}}{2K^{1/2}}\frac{\delta^{*}}{\sqrt{a-\frac{a}{2}}}\leq\left(\frac{\lambda}{2Ka}\right)^{1/2}\delta^{*1/2}
		\end{align*}
		and
		\begin{align*}
			\sum_{k=1}^{K}|\tilde{\gamma}_{lj}^{(k)}-\hat{\gamma}_{lj}^{(k)}|\leq I_{1}+I_{2},
		\end{align*}
		where
		\begin{align*}
			I_{1}&=\lambda^{-1/2}\frac{\sum_{k=1}^{K}|d_{lj}^{(k)}|}{\left(\sum_{k=1}^{K}|\hat{\theta}_{lj}(k)+d_{lj}^{(k)}|\right)^{1/2}}\leq\lambda^{-1/2}\frac{\sum_{k=1}^{K}|d_{lj}^{(k)}|}{\left(\sum_{k=1}^{K}|\hat{\theta}_{lj}(k)|-\sum_{k=1}^{K}|d_{lj}^{(k)}|\right)^{1/2}}\\
			&\leq \lambda^{-1/2}\frac{\delta^{*}}{K^{1/2}\sqrt{a-\frac{a}{2}}}\leq\left(\frac{2}{\lambda K a}\right)^{1/2}\delta^{*1/2}
		\end{align*}
		and
		\begin{align*}
			I_{2}&=\lambda^{-1/2}\sum_{k=1}^{K}\left|\frac{\hat{\theta}_{lj}(k)}{\left(\sum_{k=1}^{K}|\hat{\theta}_{lj}(k)+d_{lj}^{(k)}|\right)^{1/2}}-\frac{\hat{\theta}_{lj}(k)}{\left(\sum_{k=1}^{K}|\hat{\theta}_{lj}(k)|\right)^{1/2}}\right|\\
			&=\lambda^{-1/2}\left(\sum_{k=1}^{K}|\hat{\theta}_{lj}(k)|\right)^{1/2}\frac{\left|\left(\sum_{k=1}^{K}|\hat{\theta}_{lj}(k)+d_{lj}^{(k)}|\right)^{1/2}-\left(\sum_{k=1}^{K}|\hat{\theta}_{lj}(k)|\right)^{1/2}\right|}{\left(\sum_{k=1}^{K}|\hat{\theta}_{lj}(k)+d_{lj}^{(k)}|\right)^{1/2}}\\
			&\leq\lambda^{-1/2}\frac{\left(\sum_{k=1}^{K}|\hat{\theta}_{lj}(k)|\right)^{1/2}}{\left(\sum_{k=1}^{K}|\hat{\theta}_{lj}(k)|-\sum_{k=1}^{K}|d_{lj}^{(k)}|\right)^{1/2}}\lambda^{-1/2}|\tilde{\eta}_{lj}-\hat{\eta}_{lj}|\\
			&\leq \lambda^{-1/2}\frac{b^{1/2}}{\sqrt{Ka-\frac{Ka}{2}}}\frac{1}{\sqrt{2Ka}}\delta^{*1/2}\leq\left(\frac{b}{\lambda K^{2} a^{2}}\right)^{1/2}\delta^{*1/2}.
		\end{align*}
		Therefore,
		\begin{align*}
			||\tilde{\mathbf{H}}-\hat{\mathbf{H}}||_{1}+\sum_{k=1}^{K}||\tilde{\mathbf{\Gamma}}^{(k)}-\hat{\mathbf{\Gamma}}^{(k)}||_{1}<\delta.
		\end{align*}
		Thus,
		\begin{align*}
			Q_{2}(\lambda,\hat{\mathbf{H}},\hat{\mathbf{\Gamma}})\leq Q_{2}(\lambda,\tilde{\mathbf{H}},\tilde{\mathbf{\Gamma}})\Rightarrow Q_{3}(\lambda,\hat{\mathbf{\Theta}})\leq Q_{3}(\lambda,\tilde{\mathbf{\Theta}}),
		\end{align*}
		which means that $\hat{\mathbf{\Theta}}$ is a local minimizer of \eqref{obj:3}. The other direction can be proved similarly.	
	\end{proof}

	\medskip
	\begin{large}
		\noindent\textbf{Proof of Proposition \ref{prop:1}}
	\end{large}

	\begin{proof}
		By the Woodbury identity,
		\begin{align*}
			(\mathds{X}_{-j}^{\intercal}\mathds{X}_{-j}+nb\mathbf{I}_{p-1})^{-1}=\frac{1}{nb}\mathbf{I}_{p-1}-\frac{1}{nb}\mathds{X}_{-j}^{\intercal}(\mathds{X}_{-j}\mathds{X}_{-j}^{\intercal}+nb\mathbf{I}_{n})^{-1}\mathds{X}_{-j}.
		\end{align*}
		By the Sherman-Morrison identity, we get
		\begin{align*}
			(\mathds{X}_{-j}\mathds{X}_{-j}^{\intercal}+nb\mathbf{I}_{n})^{-1}&=(\mathds{X}\mathds{X}^{\intercal}+nb\mathbf{I}_{n}-\mathds{X}_{j}\mathds{X}_{j}^{\intercal})^{-1}\\
			&=(\mathds{X}\mathds{X}^{\intercal}+nb\mathbf{I}_{n})^{-1}+\frac{(\mathds{X}\mathds{X}^{\intercal}+nb\mathbf{I}_{n})^{-1}\mathds{X}_{j}\mathds{X}_{j}^{\intercal}(\mathds{X}\mathds{X}^{\intercal}+nb\mathbf{I}_{n})^{-1}}{1-\mathds{X}_{j}^{\intercal}(\mathds{X}\mathds{X}^{\intercal}+nb\mathbf{I}_{n})^{-1}\mathds{X}_{j}}.
		\end{align*}
		Combining the two,we have the desired result.
	\end{proof}

	\medskip
	\noindent
		\begin{large}
			\noindent\textbf{Proof of Proposition \ref{prop:2}}
		\end{large}
		\begin{proof}
			The most computationally expensive part of \eqref{admmjns} is that of computing $\mathbf{v}^{t+1}$, which consists of two components. The first component is
			\begin{align*}
				(\mathds{Z}^{\intercal}\mathds{Z}+nb\mathbf{I}_{p(p-1)})^{-1}\mathds{Z}^{\intercal}\mathds{Y}=
				\left[
				\begin{array}{c}
					(\mathds{X}_{-1}^{\intercal}\mathds{X}_{-1}+nb\mathbf{I}_{p-1})^{-1}\mathds{X}_{-1}^{\intercal}\mathds{X}_{1}\\
					\vdots\\
					(\mathds{X}_{-p}^{\intercal}\mathds{X}_{-p}+nb\mathbf{I}_{p-1})^{-1}\mathds{X}_{-p}^{\intercal}\mathds{X}_{p}
				\end{array}
				\right].
			\end{align*}
			Using proposition \ref{prop:1}, it is easy to see that the computational complexity of
			\begin{align*}
				(\mathds{X}_{-j}^{\intercal}\mathds{X}_{-j}+nb\mathbf{I}_{p-1})^{-1}\mathds{X}_{-j}^{\intercal}\mathds{X}_{j}
			\end{align*}
			is $\mathcal{O}(np)$. Therefore, the computational complexity of
			\begin{align*}
				(\mathds{Z}^{\intercal}\mathds{Z}+nb\mathbf{I}_{p(p-1)})^{-1}\mathds{Z}^{\intercal}\mathds{Y}
			\end{align*}
			is $\mathcal{O}(np^{2})$. Similarly, the computational complexity of the second component
			\begin{align*}
				(\mathds{Z}^{\intercal}\mathds{Z}+nb\mathbf{I}_{p(p-1)})^{-1}nb(\mathbf{r}^{t}-\mathbf{u}^{t})
			\end{align*}
			is $\mathcal{O}(np^{2})$, which concludes the proof.
	\end{proof}

	\medskip
	\begin{lemma}\label{asymp:lem:1}
		Suppose $\bm{\epsilon}=(\epsilon_{1},\ldots,\epsilon_{n})^{\intercal}$ follows $\mathcal{N}_{n}(\mathbf{0},\sigma^{2}\mathbf{I}_{n})$. Then, there exists a constant $M>0$ such that
		\begin{align*}
			f(t)=\sup_{||\bm{a}||_{2}\leq 1}P\left(|\bm{a}^{\intercal}\bm{\epsilon}|>t\right)\leq 2\exp\left(-M t^{2}\right)
		\end{align*}
	\end{lemma}

	\begin{proof}
		The proof of this result can be found in the supplementary material of \cite{huang2008adaptive}. However, for completeness we include it here. 
		
		For $d\in[1,+\infty)$, define $\psi_{d}(x)=\exp(x^{d})-1$. For any random variable $\mathbf{X}$, its $\psi_{d}$-Orlicz norm is defined as
		\begin{align*}
			||X||_{\psi_{d}}=\inf\left\{C\in(0,+\infty):\e\left[\psi_{d}\left(\frac{|X|}{C}\right)\right]\leq1\right\}.
		\end{align*}
		
		Since $\epsilon_{i}$ is Gaussian with mean zero and variance $\sigma^{2}$, by exercise 2.7 of \cite{boucheron2013concentration} we have
		\begin{align*}
			P(|\epsilon_{i}|>t)\leq\exp\left(-\frac{t^{2}}{2\sigma^{2}}\right).
		\end{align*}
		By Lemma 2.2.1 of \cite{van1996weak} we have $||\epsilon_{i}||_{\psi_{2}}\leq 2\sigma$. Let $\bm{a}\in\mathbb{R}^{n}$, satisfying $||\bm{a}||_{2}\leq 1$. By proposition A.1.6 of \cite{van1996weak}, there exists $K_{2}>0$ which only depends on $d$ such that
		\begin{align*}
			||\bm{a}^{\intercal}\bm{\epsilon}||_{\psi_{2}}\leq K_{2}\left[\e|\bm{a}^{\intercal}\bm{\epsilon}|+\left(\sum_{i=1}^{n}||a_{i}\epsilon_{i}||_{\psi_{2}}^{2}\right)^{1/2}\right].
		\end{align*}
		Since
		\begin{align*}
			\e(\bm{a}^{\intercal}\bm{\epsilon})=\e\left(\sum_{i=1}^{n}a_{i}^{2}\epsilon_{i}^{2}+\sum_{i\ne j}a_{i}a_{j}\epsilon_{i}\epsilon_{j}\right)=\sum_{i=1}^{n}a_{i}^{2}\sigma^{2}\leq \sigma^{2},
		\end{align*}
		by Jensen's inequality
		\begin{align*}
			\e|\bm{a}^{\intercal}\bm{\epsilon}|\leq \left[\e\left(\sum_{i=1}^{n}a_{i}\epsilon_{i}\right)^{2}\right]^{1/2}\leq \sigma.
		\end{align*}
		Concerning the second term,
		\begin{align*}
			\left(\sum_{i=1}^{n}||a_{i}\epsilon_{i}||_{\psi_{2}}^{2}\right)^{1/2}=\left(\sum_{i=1}^{n}|a_{i}|^{2}||\epsilon_{i}||_{\psi_{2}}^{2}\right)^{1/2}\leq 2\sigma||\bm{a}||_{2}\leq 2\sigma.
		\end{align*}
		Therefore, $||\bm{a}^{\intercal}\bm{\epsilon}^{(k)}||_{\psi_{2}}\leq K_{2}(\sigma^{2}+2\sigma)$.	By the definition of $||X||_{\psi_{2}}$ it is true that
		\begin{align*}
			\e\left[\exp\left(\frac{|X|^{2}}{||X||_{\psi_{2}}^{2}}\right)\right]\leq 2.
		\end{align*}
		Then, for $t>0$, with the help of Markov's inequality
		\begin{align*}
			P(|X|>t)=P\left(\exp\left\{\frac{|X|^{2}}{||X||_{\psi_{2}}^{2}}\right\}>\exp\left\{\frac{t^{2}}{||X||_{\psi_{2}}^{2}}\right\}\right)\leq 2\exp\left(\frac{t^{2}}{||X||_{\psi_{2}}^{2}}\right).
		\end{align*}
		This means that
		\begin{align*}
			P(|\bm{a}^{\intercal}\bm{\epsilon}|>t)\leq 2\exp\left(-\frac{t^{2}}{||\bm{a}^{\intercal}\bm{\epsilon}||_{\psi_{2}}^{2}}\right)\leq 2\exp\left(-\frac{t^{2}}{K_{2}^{2}(\sigma^{2}+2\sigma)^{2}}\right).
		\end{align*}
		Define $M=K_{2}^{-2}(\sigma^{2}+2\sigma)^{-2}$. Then
		\begin{align*}
			f(t)=\sup_{||\bm{a}||_{2}\leq 1}P(|\bm{a}^{\intercal}\bm{\epsilon}|>t)\leq 2\exp\left(-M t^{2}\right).
		\end{align*}
	\end{proof}

	\medskip
	\begin{lemma}\label{asymp:lem:2}
		If Assumption \ref{ass:2} is satisfied, then 
		\begin{align*}
			&P(||\mathbf{s}_{j}^{(k)}||_{2}\geq t)\leq q_{n}\exp\left\{-C\left(\frac{4t^{2}}{M_{1}q_{n}}-1\right)\right\},\quad\text{for}\quad t\geq\frac{(M_{1}q_{n})^{1/2}}{2},\\
			\text{and}\quad&P(\tau_{lj}^{-1}\geq t)\leq \exp\left\{-C\left(\frac{t^{2}}{4}-M_{2}\right)\right\},\quad\text{for}\quad t\geq 2 M_{2}^{1/2}.
		\end{align*}
	\end{lemma}
	\begin{proof}
		By the definition of $\mathbf{s}_{j}^{(k)}$,	
		\begin{align*}
			P(||\mathbf{s}_{j}^{(k)}||_{2}\geq t)&=P\left(\sum_{l=1}^{q_{nj}^{(k)}}\frac{1}{\sum_{k=1}^{K}|\tilde{\theta}_{lj}^{(k)}|}\geq 4t^{2}\right) \leq\sum_{l=1}^{q_{nj}^{(k)}}P\left(\frac{1}{\sum_{k=1}^{K}|h_{lj}^{(k)}|}\frac{\sum_{k=1}^{K}|h_{lj}^{(k)}|}{\sum_{k=1}^{K}|\tilde{\theta}_{lj}^{(k)}|}\geq \frac{4t^{2}}{q_{nj}^{(k)}}\right)\\
			&\leq \sum_{l=1}^{q_{nj}^{(k)}}P\left(M_{1}\frac{\sum_{k=1}^{K}|h_{lj}^{(k)}|}{\sum_{k=1}^{K}|\tilde{\theta}_{lj}^{(k)}|}\geq \frac{4t^{2}}{q_{n}}\right)=\sum_{l=1}^{q_{nj}^{(k)}}P\left(\frac{\sum_{k=1}^{K}|h_{lj}^{(k)}|}{\sum_{k=1}^{K}|\tilde{\theta}_{lj}^{(k)}|}-1\geq\frac{4t^{2}}{M_{1}q_{n}}-1\right)\\
			&\leq q_{n} P\left(\max_{\substack{1\leq j\leq p_{n} \\ 1\leq l\leq p_{n}-1}}\left|\frac{\sum_{k=1}^{K}|h_{lj}^{(k)}|}{\sum_{k=1}^{K}|\tilde{\theta}_{lj}^{(k)}|}-1\right|\geq\frac{4t^{2}}{M_{1}q_{n}}-1\right)\leq q_{n}\exp\left\{-C\left(\frac{4t^{2}}{M_{1}q_{n}}-1\right)\right\},
		\end{align*}
		for $t>(M_{1}q_{n})^{1/2}/2$.
		
		Also, by the definition of $\tau_{lj}$,
		\begin{align*}
			P(\tau_{lj}^{-1}\geq t)&=P(\tau_{lj}^{-2}\geq t^{2})=P\left(4\sum_{k=1}^{K}|\tilde{\theta}_{lj}^{(k)}|\geq t^{2}\right)\\
			&=P\left(\sum_{k=1}^{K}|\tilde{\theta}_{lj}^{(k)}|-\sum_{k=1}^{K}|h_{lj}^{(k)}|\geq \frac{t^{2}}{4}-\sum_{k=1}^{K}|h_{lj}^{(k)}|\right)\\
			&\leq P\left(\max_{\substack{1\leq j\leq p_{n} \\ 1\leq l\leq p_{n}-1}}\left|\sum_{k=1}^{K}|\tilde{\theta}_{lj}^{(k)}|-\sum_{k=1}^{K}|h_{lj}^{(k)}|\right| \geq \frac{t^{2}}{4}-M_{2}\right)\\
			&\leq\exp\left\{-C\left(\frac{t^{2}}{4}-M_{2}\right)\right\},
		\end{align*}
		for $t\geq 2M_{2}^{1/2}$.
	\end{proof}

	\medskip
	\begin{large}
		\noindent\textbf{Proof of Theorem \ref{thm:3}}
	\end{large}
	\begin{proof}
		Since $k$ does not depend on $n$ and the proof is the same for all $k$, we omit $k$ from this proof. Let $\hat{\bm{\theta}}_{j}$ be the unique solution of
		\begin{align}\label{obj:6}
			\argmin_{\bm{\theta}_{j}\in\mathbb{R}^{p_{n}-1}}\left(\frac{1}{2n}||\mathds{X}_{j}-\mathds{X}_{-j}\bm{\theta}_{j}||_{2}^{2}+\lambda_{n}\sum_{l=1}^{p_{n}-1}\tau_{lj}|\theta_{lj}|\right).
		\end{align}	
		Note that
		\begin{align*}
			\left\{\hat{\mathcal{E}}=\mathcal{E}\right\}\supseteq\bigcap_{j=1}^{p_{n}}\left\{\hat{\bm{\theta}}_{j}=_{s}\bm{\theta}_{j}\right\}.
		\end{align*}
		Therefore,
		\begin{align*}
			P\left(\hat{\mathcal{E}}\ne \mathcal{E}\right)\leq P\left(\bigcup_{j=1}^{p_{n}}\left\{\hat{\bm{\theta}}_{j}\ne_{s}\bm{\theta}_{j}\right\}\right)\leq\sum_{j=1}^{p_{n}}P\left(\hat{\bm{\theta}}_{j}\ne_{s}\bm{\theta}_{j}\right).
		\end{align*}
		Let $\bm{\theta}_{j}=\left(\bm{\theta}_{j;1}^{\intercal},\bm{\theta}_{j;2}^{\intercal}\right)^{\intercal}$ and $\hat{\bm{\theta}}_{j}=\left(\hat{\bm{\theta}}_{j;1}^{\intercal},\hat{\bm{\theta}}_{j;2}^{\intercal}\right)^{\intercal}$, where $\bm{\theta}_{j;1},\hat{\bm{\theta}}_{j;1}\in\mathbb{R}^{q_{nj}}$ and $\bm{\theta}_{j;2},\hat{\bm{\theta}}_{j;2}\in\mathbb{R}^{s_{nj}}$. Denote the objective function in \eqref{obj:6} by $L(\bm{\theta}_{j})$, and $\partial L(\bm{\theta}_{j})/\partial{\theta_{lj}}$ its subdifferential with respect to the $l$-th component $\theta_{lj}$ \citep{clarke1990optimization}.	For $l=1,\ldots,p_{n}-1$, 
		the KKT conditions \citep{kuhn2014nonlinear} are given by
		\begin{align*}
			0\in\frac{\partial L(\hat{\bm{\theta}}_{j})}{\partial\theta_{lj}}\Leftrightarrow
			\begin{cases}
				(\mathds{X}_{-j})_{l}^{\intercal}(\mathds{X}_{j}-\mathds{X}_{-j}\hat{\bm{\theta}}_{j})=n\lambda_{n}\tau_{lj}\sgn(\hat{\theta}_{lj}),&\hat{\theta}_{lj}\ne 0\\
				\left|(\mathds{X}_{-j})_{l}^{\intercal}(\mathds{X}_{j}-\mathds{X}_{-j}\hat{\bm{\theta}}_{j})\right|\leq n\lambda_{n}\tau_{lj},&\hat{\theta}_{lj}=0
			\end{cases}
			.
		\end{align*}
		Thus
		\begin{align*}
			\hat{\bm{\theta}}_{j}=_{s}\bm{\theta}_{j}\Rightarrow
			\begin{cases}
				\mathds{X}_{-j;1}^{\intercal}[\mathds{X}_{j}-\mathds{X}_{-j;1}\hat{\bm{\theta}}_{j;1}]=n\lambda_{n}\mathbf{s}_{j}\\
				\left|(\mathds{X}_{-j})_{l}^{\intercal}[\mathds{X}_{j}-\mathds{X}_{-j;1}\hat{\bm{\theta}}_{j;1}]\right|\leq n\lambda_{n}\tau_{lj},&l=q_{nj}+1,\ldots,p_{n}-1
			\end{cases}
			.
		\end{align*}
		Define the matrix
		\begin{align*}
			\hat{\mathbf{H}}_{jj}=\mathbf{I}_{n}-n^{-1}\mathds{X}_{-j;1}\hat{\mathbf{\Sigma}}_{jj}^{-1}\mathds{X}_{-j;1}^{\intercal},
		\end{align*}
		and let $\bm{\epsilon}=(\epsilon_{1},\ldots,\epsilon_{n})^{\intercal}$ be the vector of errors mentioned in section 2. Then,
		\begin{align*}
			\mathds{X}_{j}&=\mathds{X}_{-j;1}\bm{\theta}_{j;1}+\bm{\epsilon}\\
			\Rightarrow \mathds{X}_{-j;1}^{\intercal}\mathds{X}_{j}&=\mathds{X}_{-j;1}^{\intercal}\mathds{X}_{-j;1}\bm{\theta}_{j;1}+\mathds{X}_{-j;1}^{\intercal}\bm{\epsilon},
		\end{align*}
		and
		\begin{align*}
			\hat{\bm{\theta}}_{j}=_{s}\bm{\theta}_{j}\Rightarrow
			\begin{cases}
				\hat{\bm{\theta}}_{j;1}=\bm{\theta}_{j;1}+\frac{1}{n}\hat{\mathbf{\Sigma}}_{jj}^{-1}\left(\mathds{X}_{-j;1}^{\intercal}\bm{\epsilon}-n\lambda_{n}\mathbf{s}_{j}\right)\\
				\left|(\mathds{X}_{-j})_{l}^{\intercal}\left[\hat{\mathbf{H}}_{jj}\bm{\epsilon}+\lambda_{n}\mathds{X}_{-j;1}\hat{\mathbf{\Sigma}}_{jj}^{-1}\mathbf{s}_{j}\right]\right|\leq n\lambda_{n}\tau_{lj},&l=q_{nj}+1,\ldots,p_{n}-1
			\end{cases}
			.
		\end{align*}
		Let $\mathbf{e}_{l}$ be the unit vector in the direction of the $l$-th coordinate. Then,
		\begin{align*}
			\hat{\theta}_{lj}=\theta_{lj}+\frac{1}{n}\mathbf{e}_{l}^{\intercal}\hat{\mathbf{\Sigma}}_{jj}^{-1}\left(\mathds{X}_{-j;1}^{\intercal}\bm{\epsilon}-n\lambda_{n}\mathbf{s}_{j}\right).
		\end{align*}
		For $l=1,\ldots,q_{nj}$, a necessary condition for $\sgn(\hat{\theta}_{lj})\ne \sgn(\theta_{lj})$ is
		\begin{align*}
			|\theta_{lj}|<\frac{1}{n}\left|\mathbf{e}_{l}^{\intercal}\hat{\mathbf{\Sigma}}_{jj}^{-1}\left(\mathds{X}_{-j;1}^{\intercal}\bm{\epsilon}-n\lambda_{n}\mathbf{s}_{j}\right)\right|.
		\end{align*}
		For $l=q_{nj}+1,\ldots,p_{n}-1$, a necessary condition for $\hat{\theta}_{lj}\ne 0$ is
		\begin{align*}
			\left|(\mathds{X}_{-j})_{l}^{\intercal}\left(\hat{\mathbf{H}}_{jj}\bm{\epsilon}+\lambda_{n}\mathds{X}_{-j;1}\hat{\mathbf{\Sigma}}_{jj}^{-1}\mathbf{s}_{j}\right)\right|> n\lambda_{n}\tau_{lj}.
		\end{align*}
		Therefore,
		\begin{align*}
			P(\hat{\bm{\theta}}_{j}\ne_{s} \bm{\theta}_{j})&\leq P\left(\bigcup_{l=1}^{q_{nj}}\left\{\frac{1}{n}\left|\mathbf{e}_{l}^{\intercal}\hat{\mathbf{\Sigma}}_{jj}^{-1}\left(\mathds{X}_{-j;1}^{\intercal}\bm{\epsilon}-n\lambda_{n}\mathbf{s}_{j}\right)\right|>|\theta_{lj}|\right\}\right)\\
			&+P\left(\bigcup_{l=q_{nj}+1}^{p_{n}-1}\left\{\left|(\mathds{X}_{-j})_{l}^{\intercal}\left(\hat{\mathbf{H}}_{jj}\bm{\epsilon}+\lambda_{n}\mathds{X}_{-j;1}\hat{\mathbf{\Sigma}}_{jj}^{-1}\mathbf{s}_{j}\right)\right|> n\lambda_{n}\tau_{lj}\right\}\right)\\
			&\leq P(B_{1})+P(B_{2})+P(B_{3})+P(B_{4}),
		\end{align*}
		where
		\begin{align*}
			P(B_{1})&=P\left(\bigcup_{l=1}^{q_{nj}}\left\{\frac{1}{n}\left|\mathbf{e}_{l}^{\intercal}\hat{\mathbf{\Sigma}}_{jj}^{-1}\mathds{X}_{-j;1}^{\intercal}\bm{\epsilon}\right|>\frac{|\theta_{lj}|}{2}\right\}\right)\\
			P(B_{2})&=P\left(\bigcup_{l=1}^{q_{nj}}\left\{\lambda_{n}\left|\mathbf{e}_{l}^{\intercal}\hat{\mathbf{\Sigma}}_{jj}^{-1}\mathbf{s}_{j}\right|>\frac{|\theta_{lj}|}{2}\right\}\right)\\
			P(B_{3})&=P\left(\bigcup_{l=q_{nj}+1}^{p_{n}-1}\left\{\left|(\mathds{X}_{-j})_{l}^{\intercal}\hat{\mathbf{H}}_{jj}\bm{\epsilon}\right|>\frac{n\lambda_{n}\tau_{lj}}{2}\right\}\right)\\
			P(B_{4})&=P\left(\bigcup_{l=q_{nj}+1}^{p_{n}-1}\left\{\left|(\mathds{X}_{-j})_{l}^{\intercal}\mathds{X}_{-j;1}\hat{\mathbf{\Sigma}}_{jj}^{-1}\mathbf{s}_{j}\right|>\frac{n\tau_{lj}}{2}\right\}\right)
		\end{align*}
		
		\bigskip
		Concerning $P(B_{1})$, using Assumption \ref{ass:1} we get
		\begin{align*}
			\frac{1}{n^{2}}\left|\left|\mathbf{e}_{l}^{\intercal}\hat{\mathbf{\Sigma}}_{jj}^{-1}\mathds{X}_{-j;1}^{\intercal}\right|\right|_{2}^{2}=\frac{1}{n}\mathbf{e}_{l}^{\intercal}\hat{\mathbf{\Sigma}}_{jj}^{-1}\mathbf{e}_{l}\leq \frac{1}{n}\upsilon_{nj}^{-1}\leq \frac{1}{n}\xi^{-1}.
		\end{align*}
		Therefore,
		\begin{align*}
			\left|\left|\left(\frac{\xi}{n}\right)^{1/2}\mathbf{e}_{l}^{\intercal}\hat{\mathbf{\Sigma}}_{jj}^{-1}\mathds{X}_{-j;1}^{\intercal}\right|\right|_{2}\leq 1,
		\end{align*}
		and with the help of Lemma \ref{asymp:lem:1},
		\begin{align*}
			P(B_{1})&\leq\sum_{l=1}^{q_{nj}}P\left(\left|\left(\frac{\xi}{n}\right)^{1/2}\mathbf{e}_{l}^{\intercal}\hat{\mathbf{\Sigma}}_{jj}^{-1}\mathds{X}_{-j;1}^{\intercal}\bm{\epsilon}\right|\geq \frac{(n\xi)^{1/2}}{2}|\theta_{lj}| \right)\\
			&\leq \sum_{l=1}^{q_{nj}}P\left(\left|\left(\frac{\xi}{n}\right)^{1/2}\mathbf{e}_{l}^{\intercal}\hat{\mathbf{\Sigma}}_{jj}^{-1}\mathds{X}_{-j;1}^{\intercal}\bm{\epsilon}\right|\geq \frac{(n\xi)^{1/2}}{2}b_{n}\right)\\
			&\leq \sum_{l=1}^{q_{nj}}f(\sqrt{n\xi}b_{n}/2)\leq q_{n}f(\sqrt{n\xi}b_{n}/2)\\
			&\leq q_{n}\exp\left\{-\frac{M \xi}{4}nb_{n}^{2}\right\}.
		\end{align*}

		\bigskip
		Concerning $P(B_{2})$, by the Cauchy-Schwarz inequality,
		\begin{align*}
			\lambda_{n}\left|\mathbf{e}_{l}^{\intercal}\hat{\mathbf{\Sigma}}_{jj}^{-1}\mathbf{s}_{j}\right|\leq \lambda_{n}\left|\left|\mathbf{e}_{l}^{\intercal}\hat{\mathbf{\Sigma}}_{jj}^{-1}\right|\right|_{2}||\mathbf{s}_{j}||_{2}\leq \frac{\lambda_{n}||\mathbf{s}_{j}||_{2}}{\xi}.
		\end{align*}
		Thus, by Lemma \ref{asymp:lem:2}, with $t=\xi b_{n}/2\lambda_{n}$,
		\begin{align*}
			&P(B_{2})\leq P\left(||\mathbf{s}_{j}||_{2}\geq \frac{\xi}{2}\frac{b_{n}}{\lambda_{n}}\right)\leq q_{n}\exp\left\{-C\left(\frac{\xi^{2}}{M_{1}}\frac{b_{n}^{2}}{\lambda_{n}^{2}q_{n}}-1\right)\right\}.
		\end{align*}
		Note that by Assumption \ref{ass:4}, the condition
		\begin{align*}
			\frac{\xi b_{n}}{2\lambda_{n}}>\frac{(M_{1}q_{n})^{1/2}}{2},
		\end{align*}
		is satisfied for sufficiently large $n$.

		\bigskip
		Concerning $P(B_{3})$, since $\hat{\mathbf{H}}_{jj}$ is a projection, its operator norm is bounded by 1. Thus, by \eqref{centerscale},
		\begin{align*}
			||(\mathds{X}_{-j})_{l}^{\intercal}\hat{\mathbf{H}}_{jj}||_{2}\leq ||(\mathds{X}_{-j})_{l}||_{2}=n^{1/2}\Rightarrow ||n^{-1/2}(\mathds{X}_{-j})_{l}^{\intercal}\hat{\mathbf{H}}_{jj}||_{2}\leq 1.
		\end{align*}
		Therefore, with the help of Lemma \ref{asymp:lem:1} and Lemma \ref{asymp:lem:2}, for sufficiently large $n$, we have
		\begin{align*}
			P(B_{3})&\leq\sum_{l=q_{nj}+1}^{p_{n}-1}P\left(\left|(\mathds{X}_{-j})_{l}^{\intercal}\hat{\mathbf{H}}_{jj}\bm{\epsilon}\right|>\frac{n\lambda_{n}\tau_{lj}}{2}\right)\\
			&\leq \sum_{l=q_{nj}+1}^{p_{n}-1}P(\tau_{lj}^{-1}\geq n^{1/2})+\sum_{l=q_{nj}+1}^{p_{n}-1}P\left(\left|(\mathds{X}_{-j})_{l}^{\intercal}\hat{\mathbf{H}}_{jj}\bm{\epsilon}\right|\geq\frac{n^{1/2}\lambda_{n}}{2}\right)\\
			&= \sum_{l=q_{nj}+1}^{p_{n}-1}P(\tau_{lj}^{-1}\geq n^{1/2})+\sum_{l=q_{nj}+1}^{p_{n}-1}P\left(\left|\frac{1}{n^{1/2}}(\mathds{X}_{-j})_{l}^{\intercal}\hat{\mathbf{H}}_{jj}\bm{\epsilon}\right|\geq\frac{\lambda_{n}}{2}\right)\\
			&\leq s_{n}\exp\left\{-C\left(\frac{n}{4}-M_{2}\right)\right\}+s_{n}\exp\left\{-\frac{M}{4}\lambda_{n}^{2}\right\}.
		\end{align*}

		\bigskip
		Concerning $P(B_{4})$, with the help of Lemma \ref{asymp:lem:2} and Assumption \ref{ass:3}, for sufficiently large $n$, we have,
		\begin{align*}
			P(B_{4})&\leq \sum_{l=q_{nj}+1}^{p_{n}-1}P(\tau_{lj}^{-1}\geq n^{1/2})+\sum_{l=q_{nj}+1}^{p_{n}-1}P\left(\left|(\mathds{X}_{-j})_{l}^{\intercal}\mathds{X}_{-j;1}\hat{\mathbf{\Sigma}}_{jj}^{-1}\mathbf{s}_{j}\right|\geq\frac{n^{1/2}}{2}\right)\\
			&\leq s_{n}\exp\left\{-C\left(\frac{n}{4}-M_{2}\right)\right\}+s_{n}\exp\left\{-\frac{C}{4}n\right\}.
		\end{align*}
		
		\bigskip
		Combining the upper bounds for $P(B_{1}),P(B_{2}),P(B_{3}),P(B_{4})$, we have, for sufficiently large $n$,
			\begin{align*}
				\sum_{j=1}^{p_{n}}P\left(\hat{\bm{\theta}}_{j}\ne_{s}\bm{\theta}_{j}\right)&\leq \exp\left\{\log(p_{n})\left[2-\frac{M\xi}{4}\frac{nb_{n}^{2}}{\log p_{n}}\right]\right\}\\
				&+\exp\left\{\log(p_{n})\left[2-\frac{C\xi^{2}}{M_{1}}\frac{b_{n}^{2}}{\lambda_{n}^{2}q_{n}\log p_{n}}+\frac{C}{\log p_{n}}\right]\right\}\\
				&+\exp\left\{\log(p_{n})\left[2-\frac{M}{4}\frac{\lambda_{n}^{2}}{\log p_{n}}\right]\right\}\\
				&+\exp\left\{\log(p_{n})\left[2-\frac{C}{4}\frac{n}{\log p_{n}}\right]\right\}\\
				&+2\exp\left\{\log(p_{n})\left[2-\frac{C}{4}\frac{n}{\log p_{n}}+\frac{CM_{2}}{\log p_{n}}\right]\right\},
			\end{align*}
			where the right-hand side converges to 0 by Assumption \ref{ass:4}.
	\end{proof}

	\newpage
	\nocite{*}
	\bibliographystyle{asa}
	\bibliography{mybibliography}	
\end{document}